\documentclass[prd,11pt,tightenlines,nofootinbib,superscriptaddress]{revtex4}
\usepackage{amsfonts,amssymb,amsthm,bbm,amsmath}
\usepackage{verbatim}
\usepackage{graphicx}

\newcommand{\C}{{\mathbb C}}

\newcommand{\R}{{\mathbb R}}

\newcommand{\bbT}{{\mathbb T}}
\newcommand{\one}{\mathbbm{1}}

\newcommand{\cH}{{\mathcal H}}

\newcommand{\SU}{\mathrm{SU}}

\newcommand{\SL}{\mathrm{SL}}

\newcommand{\be}{\begin{equation}}
\newcommand{\ee}{\end{equation}}
\newcommand{\beq}{\begin{eqnarray}}
\newcommand{\eeq}{\end{eqnarray}}
\newcommand{\bea}{\begin{eqnarray}}
\newcommand{\eea}{\end{eqnarray}}
\newcommand{\nn}{\nonumber}

\newcommand{\bra}{\langle}
\newcommand{\ket}{\rangle}
\newcommand{\la}{\langle}
\newcommand{\ra}{\rangle}

\newcommand{\tr}{{\mathrm Tr}}

\newcommand{\rd}{\mathrm{d}}

\newcommand{\bpm}{\begin{pmatrix}}
\newcommand{\epm}{\end{pmatrix}}
\newcommand{\bvm}{\begin{vmatrix}}
\newcommand{\evm}{\end{vmatrix}}

\def\nn{\nonumber}

\newcommand{\corner}{\resizebox{15pt}{!}{\mbox{\huge{$\lrcorner$}}}}

\newtheorem{theorem}{Theorem}[section]
\newtheorem{lemma}[theorem]{Lemma}
\newtheorem{proposition}[theorem]{Proposition}

\newcommand{\sixj}[6]{
\left\{
\begin{array}{ccc}
#1 & #2 & #3 \\
#4 & #5 & #6
\end{array}
\right\}
}

\newcommand{\threej}[6]{
\left(
\begin{array}{ccc}
#1 & #2 & #3 \\
#4 & #5 & #6
\end{array}
\right)
}

\begin{document}

\title{A Discrete and Coherent Basis of Intertwiners}

\author{{\bf Laurent Freidel}\email{lfreidel@perimeterinsititute.ca}}
\affiliation{ Perimeter Institute for Theoretical Physics,
Waterloo, Ontario, Canada.}

\author{{\bf Jeff Hnybida}\email{jhnybida@perimeterinsititute.ca}}
\affiliation{ Perimeter Institute for Theoretical Physics,
Waterloo, Ontario, Canada.}
\affiliation{Department of Physics, University of Waterloo, Waterloo, Ontario, Canada}

\begin{abstract}
We construct a new discrete basis of $4$-valent SU(2) intertwiners.
This basis possesses both the advantage of being discrete, while at the same time representing accurately the classical degrees of freedom; hence it is coherent. The closed spin network amplitude obtained from these intertwiners depends on twenty spins and can be evaluated by a generalization of the Racah formula for an arbitrary graph. The asymptotic limit of these amplitudes is found. We give, for the first time, the asymptotics of 15j symbols in the real basis. Remarkably it gives a generalization of the Regge action to twisted geometries.
%
%
% which is labelled by two extra spins and is hence overcomplete. The basis has a simple resolution of identity and is related to the orthonormal bases by simply summing over one of the extra spins.   We show that the semi-classical limit of these states corresponds uniquely (up to gauge transformations) to a classical framed tetrahedron. 
%This basis possess both the advantage of being discrete and while at the same time represents accurately the classical degree of freedom, hence coherent.
%It also lead to the possibility to generalize Racah formula for the spin network evaluation based one 
\end{abstract}

\maketitle

%%%%%%%%%%%%%%%%%%%%%%%%%%%%%%%%%%%%%%%%%%%%%%%%%%%%%%%%%%%%%%%%%%%%%%%%%%%%%%%%%%%%%%%

%%%%%%%%%%%%%%%%%
\section{Introduction}
%%%%%%%%%%%%%%%%%

It has been observed long ago that the compositions of quantum states of angular momentum are related to geometrical objects \cite{Schwinger, wigner, Bargmann}.  The simplest example is the Clebsch-Gordan coefficients which vanish unless the spins satisfy the so called triangle relations.  A less trivial example is the Wigner 6j symbol which vanishes unless the spins can represent the edge lengths of a tetrahedron.  This insight was one of the motivations which led Ponzano and Regge to use the 6j symbol as the building block for a theory of quantum gravity in three dimensions \cite{Regge:1961px,PR}, together with the fact that the asymptotic limit of the 6j symbol is related to the discretized version of the Einstein-Hilbert action.  In higher dimensions this line of thought led to the idea of  spin foam models which are a generalization of   Ponzano and Regge's idea to a four dimensional model of quantum General Relativity.  For a review see \cite{Perez:2012wv}.  

Following a canonical approach, Loop Quantum Gravity came to the same conclusion: Geometrical quantities such area and volume are quantized \cite{Rovelli:1994ge}.  In fact there are many remarkable similarities between LQG and spin foam models.  For instance, the building blocks of both models are SU(2) intertwiners which represent quanta of space \cite{Barbieri:1997ks,Baez:1999tk, Rovelli:2006fw}.  An intertwiner is simply an invariant tensor on the group.  In other words if we denote by $V^{j}$ the $2j+1$ dimensional vector space representing spin $j$ then an intertwiner is an element of the SU(2) invariant subspace of this tensor product which we will denote
\be \label{eqn_inter_space}
  \cH_{j_1,...,j_n} \equiv \text{Inv}_{\text{SU}(2)}\left[V^{j_1} \otimes \cdots \otimes V^{j_n} \right].
\ee

The vectors in this Hilbert space will be referred to as $n$-valent intertwiners since they are represented graphically by an $n$-valent node.  The legs of this node carry the spins $j_i$ which can be interpreted as the areas of the faces of a polyhedron \cite{OH,LFetera1,LFetera2,Bianchi:2010gc} which is a consequence of the celebrated Guillemin-Sternberg theorem \cite{GS}.  In this paper we will be focused on 4-valent intertwiners but many of the methods developed here can be extended to the $n$-valent case.

Two types of basis for the space (\ref{eqn_inter_space}) are usually considered. The most common one is a discrete basis 
which diagonalizes a commuting set of  operators invariant under adjoint action.  Each of these operators is associated with a decomposition of an $n$-valent vertex as a contraction of 3-valent ones and the basis elements are constructed by the contraction of 3-valent intertwiners along the corresponding channels.
In the 4-valent case the different bases  correspond to the $S,T, U$ channels and are labelled by one extra spin.
 Bases constructed in this way are orthonormal, but lack sufficient data to describe a classical geometry.  For example a tetrahedron is uniquely determined by six quantities, such as the four areas and two of the dihedral angles between faces.  Therefore the orthonormal basis of 4-valent intertwiners, having only five labels, is not suitable for the specification of a fixed tetrahedral geometry.

Another basis usually considered is the coherent state basis (introduced in Loop gravity by Livine Speziale  \cite{coh1})  which 
overcomes this difficulty by introducing an overcomplete basis for ${\cal H}_{j_{1}\cdots j_{n}}$ which is labelled by a set of normal vectors.  In the semiclassical limit these normal vectors satisfy the closure relation  and can be identified with the vectors normal to the faces of a tetrahedron whose areas are proportional to the spins $j_{i}$ \cite{Conrady:2009px,Freidel:2009nu}.  These insights have since led to a reformulation of LQG in terms of twisted geometries \cite{twisted} obtained by matching the normal vectors corresponding to faces of classical polyhedra, but such that the shapes of the glued faces do not necessarily match.
One of the drawbacks of this basis, however, is that it introduces a continuum of states to represent a simple finite dimensional Hilbert space.  Moreover, in this representation the link with the real basis and its simplicity is lost.

What we want to investigate is whether there exists a basis which captures the essence of both constructions by being at the same time discrete and coherent. 
We provide here such a construction for the Hilbert spaces ${\cal H}_{j_{1}\cdots j_{n}}$ and focus our description to the 4-valent case, which exemplifies the main features of this basis.
We will also show that it leads to a simple generalization of the Racah formula when we use these intertwiners in the spin-network evaluation of a general graph. 
This basis arises naturally in the expansion of generating functionals for spin networks which were provided  in \cite{FH_exact}.  Finally we present a new $\{20j\}$ symbol corresponding to the coherent spin network amplitude on a 4-simplex which is a simple generalization of the $\{15j\}$ symbol.

We find that this new basis of intertwiners is overcomplete, but is labelled by a set of discrete labels.  Further, we show that the discrete basis generates the different orthonormal bases by simply summing over the extra labels.  In this way many of the different variations of recoupling coefficients in the orthonormal bases can be generated from these more fundamental amplitudes.  For example the different versions of the $\{15j\}$ symbol can be obtained by summing over five spins of the $\{20j\}$ symbol.

Finally we show that the semiclassical limit of the discrete basis corresponds to a classical framed tetrahedron in the same way as the coherent intertwiners.  A framed tetrahedron is a tetrahedron with a unit vector in each of the faces representing a 2d frame.  This frame vector is determined by the phase of the spinor and encodes the extrinsic geometry of the triangulation.  The asymptotics of the $\{20j\}$ symbol imply a classical action which, under certain geometricity constraints, is found to agree with the Regge action, which agrees with the analysis in \cite{Dittrich:2008ar} and \cite{Dittrich:2010ey}.  When these constraints are not imposed the geometry is twisted in the sense of \cite{twisted}.
 
The paper is organized as follows.  First we define the new basis of $n$-valent intertwiners in the holomorphic representation and we describe the 3-valent case.  Next we investigate the 4-valent case and construct the resolution of identity.  We then compute the action of the invariant operators and show how the new basis relates to the orthonormal bases.  Using generating functional techniques we compute the scalar product in the new basis, and use the relations with the orthogonal basis to generate the various other possible scalar products.  We then discuss the utility of this basis in representing and computing spin network amplitudes and introduce the $\{20j\}$ symbol in the 4-valent case.  Finally we study the semi-classical behaviour of the new states and we show that they correspond uniquely to classical framed tetrahedra.  Moreover, the asymptotics of the $\{20j\}$ symbol is shown to admit an interpretation as a generalization of the Regge action to twisted geometry.

\section{The New Basis}

One particularly useful representation of SU(2) is the so called Bargmann-Fock or holomorphic representation \cite{Bargmann,Schwinger}.  This space consists of holomorphic functions on spinor space $\C^2$ endowed with the Hermitian inner product
\be \label{barg_in_prod}
  \bra f | g \ket = \int_{\C^2} \overline{f(z)} g(z) \rd\mu(z)
\ee
where $\rd\mu(z) = \pi^{-2} e^{-\bra z | z \ket}  \rd^{4}z$ and $\rd^{4}z$ is the Lebesgue measure on $\C^2$
and we use the bra-ket notation for the scalar product of the two states $|f\ket, |g\ket$ defined by
$ f(z) \equiv (z|f\ket$.  The group SU(2) acts irreducibly on representations of spin $j$ given by the 
$2j+1$ dimensional subspaces of holomorphic functions homogeneous of degree $2j$.  
The standard orthonormal basis with respect to this inner product is given by
\be
   e^{j}_{m}(z) = \frac{\alpha^{j+m}\beta^{j-m}}{\sqrt{(j+m)!(j-m)!}},
\ee
which are simply the holomorphic represention of the SU(2) basis elements which diagonalize the operator $J_3$.  Here $(\alpha, \beta)\in \C^{2}$ represents the components of the spinor $|z\ket$.

In the following we will heavily use the fact that there are  two SU(2) invariant products on spinor space, only one of which is holomorphic:\footnote{In fact because $[ z_i | z_j \ket$ is holomoprhic it is automatically $\SL(2,\C)$ invariant.  The other SU(2) invariant is $\bra z_i | z_j \ket = \bar{\alpha}_i \beta_j + \alpha_j \bar{\beta}_{i}$ }:
\be
  [z_i|z_j\ket = \alpha_{i}\beta_{j} - \alpha_{j}\beta_{i}
\ee
where we use the notation
 $$|z\ket \equiv (\alpha, \beta )^t, \qquad |z] \equiv ( -\overline{\beta}, \overline{\alpha} )^t.$$
 We will also use the notation $\check{z}$ to denote the conjugate spinor $|\check{z}\ket \equiv |z]$.

In the spin $j$ representation we define coherent states $|j,z\ket$ to be the holomorphic functionals
\be
(w|j,z\ket \equiv \frac{[w|z\ket^{2j}}{(2j)!}.
\ee
These states possess the characteristic property that their scalar product with any spin $j$ state $|f\ket$ reproduces the functional $f(z)$, that is
\be
\bra j, \check{z} | f\ket = f(z).
\ee
This follows from a  direct computation which shows that
$
(w|j,z\ket =\bra j,\check{w}|j,z\ket.
$
This property implies that we can identify the label $(z|$ of $(z|f\ket$ with the state $\bra j,z|$ when 
evaluated on a spin $j$ functional. In the following we will use interchangeably the notation $(z|=\bra j,\check{z}|$ for the labels.

\subsection{Intertwiner Bases}

In this holomorphic representation there are two natural and straightforward bases of  $n$-valent intertwiners, i.e. functions of the spinors $z_{1},\cdots ,z_{n}$ which are SU(2) invariant and homogeneous of degree $2j_{i}$ in $z_{i}$.
The first one is the Livine-Speziale coherent intertwiner basis \cite{coh1}, and 
the second one is the discrete basis which is the new basis we want to study here.

The Livine-Speziale coherent intertwiners \cite{coh1} are defined by group averaging as the following holomorphic functionals:
\be
 (w_{i} \|j_i,z_i\ket \equiv \int \rd g \prod_{i=1}^{n} \frac{ [w_{i}| g |z_i\ket^{ 2j_i}}{(2j_i)!}.
\ee
Note that the normalization of these states is different from \cite{coh1} to better suit the Bargmann scalar product (\ref{barg_in_prod}).  These states are coherent in the sense that their scalar product  reproduces the
holomorphic functional, they are labelled by the continuous set of data $\{z_i\}$ and  they
resolve  the identity:
\be
  \bra j_{i}, w_{i} \|j_i,z_i\ket=(\check{w}_{i} \|j_i,z_i\ket,\qquad \one_{j_i} = \int \prod_i \rd\mu(z_i) \|j_i,z_i \ket \bra j_i,z_i\|. \label{coherent}
\ee
This is shown by using the identity $\int \rd\mu(w) = \bra a|w\ket^{2j}\bra w|b\ket^{2j} = (2j)! \bra a | b \ket^{2j}$, which itself is proven by summing over $j$ and performing the Gaussian integration.

We will now show how to construct a new basis which is also coherent, resolves the identity, but is labelled by a discrete set.
Since the product $[z|w\ket$ is holomorphic and SU(2) invariant it can be used to construct a  complete basis of the intertwiner space ${\cal H}_{n}\equiv \oplus_{j_{i}}{\cal H}_{j_{1}\cdots j_{n}}$  by
\be\label{C}
( z_{i} | k_{ij} \ket  \equiv \prod_{i<j}\frac{ [z_{i}|z_{j}\ket^{k_{ij}}}{k_{ij}!}.
\ee
This basis  is labelled by  $n(n-1)/2$ non-negative integers $[k]\equiv (k_{ij})_{i\neq j = 1,\cdots, n}$ with $k_{ij}=k_{ji}$.  
Note that we are free to choose a phase convention and for simplicity we will choose it to be unity for now.\footnote{Later we will see that the asymptotic limit of the intertwiners will imply a canonical phase.}

For a basis representing the subspace ${\cal H}_{j_{1}\cdots j_{n}}$ with fixed  spins $j_i$, we have $n$ homogeneity conditions which require the integers $[k]$ to satisfy
\be\label{kj}
\sum_{j\neq i} k_{ij} =2j_{i}.
\ee
The sum of spins at the vertex is defined by $J = \sum_i j_i = \sum_{i<j} k_{ij}$ and is required to be a positive integer.  %As shown in \cite{FH_exact} this basis has a simple resolution of identity on the space of $n$-valent intertwiners which is given by
%\be
%  \one_j = \sum_{[k] \in K_j} \frac{\left|C_{[k]}^{(n)}\right\ket \left\bra C_{[k]}^{(n)}\right|}{\Delta^2(k_{ij})},  \hspace{25pt} \Delta^2(k_{ij}) \equiv \frac{(J+1)!}{\prod_{i<j}k_{ij}!}.
%\ee
%where we define the states by $\left\bra C^{(n)}_{[k']} \right| z_i \Big\ket = C_{[k]}^{(n)}(z_{i})$ and $\Delta^2(k_{ij})$ is the $n$-valent generalization of the triangle coefficients, which in the $n=3$ case are normalization coefficients.  In the next section we will give a simplified proof of this for the 4-valent case.  
From the relation $[\check{w}|\check{z}\ket = \overline{[w|z\ket}$ we see that these states satisfy the reality condition 
\be\label{real}
\overline{( z_{i} | k_{ij} \ket} = ( \check{z}_{i} | k_{ij} \ket.
\ee
Furthermore, from the coherency property (\ref{coherent}) we can easily compute the overlap of these states with the coherent intertwiners:
\be
\bra j_{i}, \check{z}_{i} || k_{ij} \ket = ( z_{i} | k_{ij} \ket = \bra k_{ij} || j, z_{i} \ket .
\ee
where the last equality follows from 
the reality condition (\ref{real}) and the fact that $(-z_{i}|k_{ij}\ket = (z_{i}||k_{ij}\ket$.

In \cite{FH_exact} it is shown that the scalar product of 
coherent intertwiners can be expressed in terms of the coefficients of the discrete basis as 
\be\label{fundrelation}
\bra j_{i}, \check{w}_{i} || j_{i}, z_{i}\ket = \sum_{[k]\in K_{j}} \frac{( {w}_{i} | k_{ij} \ket {( z_{i} | k_{ij} \ket} }{||[k]||^{2}},\quad \mathrm{with}\quad ||[k]||^{2} =\frac{ (J+1)!}{\prod_{i<j}k_{ij}!}.
\ee
where $K_{j}$ denotes all the $k_{ij}$ solution of (\ref{kj}).
This result in turn implies that 
\be
|| j_{i}, z_{i}\ket = \sum_{[k]\in K_{j}} \frac{| k_{ij} \ket \bra k_{ij} \| j, z_{i} \ket }{||[k]||^{2}},
\ee
which expresses the coherent states in terms of the discrete basis.

\subsection{3-valent Intertwiners}

In the case $n=3$ there is only one intertwiner.  Indeed, given $[k]=(k_{12},k_{23},k_{31})$ the  homogeneity restriction requires $2j_{1}=k_{12} + k_{13}$ which can be easily solved by
\be \label{eqn_3_k}
k_{12} = j_1 + j_2 - j_3,\qquad k_{13} = j_1 - j_2 + j_3, \qquad k_{23} = -j_1 + j_2 + j_3.
\ee
In this case the fact that homogeneous functions of different degree are orthogonal implies that $|{k_{12},k_{23},k_{31}}\ket$ form an orthogonal basis
\footnote{One can also arrive at this basis by considering the respresentation space of symmetrized spinors.  For details see appendix A of \cite{Rovelli:2004tv}.  The two approaches are essentially the same, however in the holomorphic representation we have the advantage of tools like generating functionals and Gaussian integration.} of (\ref{eqn_inter_space}).

Since there is only one holomorphic  function $( z_{i}|{[k]}\ket $ it must be proportional to the Wigner 3j symbol
\be \label{C3}
 ( z_{i}| k_{12},k_{23},k_{31} \ket= \Delta(j_1 j_2 j_3) \sum_{m_1 m_2 m_3} \threej{j_1}{j_2}{j_3}{m_1}{m_2}{m_3} e^{j_1}_{m_1}(z_1) e^{j_2}_{m_2}(z_2) e^{j_3}_{m_3}(z_3)
\ee
where the triangle coefficients can be found to be 
\be \label{eqn_tri_coeff}
\Delta^{2}(j_{1}j_{2}j_{3}) \equiv 
%\left\bra C_{j_{1}j_{2}j_{3}}\right|\left. C_{j_{1}j_{2}j_{3}}\right\ket=
 \frac{(j_{1}+j_{2}+j_{3}+1)!}{(j_{1}-j_{2}+j_{3})! (j_{2}-j_{1}+j_{3})!(j_{1}+j_{2}-j_{3})!}.
\ee 
Note that we could divide $| k_{12},k_{23},k_{31}\ket$ by $\Delta(j_1 j_2 j_3)$ to normalize this basis, but it will be simpler to instead work with these unnormalized states.

\subsection{Counting}

For $n>3$ there are more basis elements $|{k_{ij}}\ket$ than the dimension of the intertwined space so the basis is no longer orthogonal.  Indeed, since we have $n(n-1)/2$ $k_{ij}$'s satisfying $n$ relations (\ref{kj}) these intertwiners are labelled by 
$n(n-3)/2$ integers. But this is clearly more that the dimension of the Hilbert space of $n$-valent intertwiners, which is known to be labelled by $n-3$ integers, i.e. by contracting only 3-valent nodes.
This means that the basis given above is { \it overcomplete}.

Another way to understand this counting is to recall that the algebra of gauge invariant operators acting on ${\cal H}_{j_{1},\cdots, j_{n}}$ is given by $J_{ij} \equiv J_{i}\cdot J_{j}$ for $i\neq j$ where $J_{i}$ denotes the angular momentum operator action in the $i$ direction.
These operators satisfy the closure relation $\sum_{i} J_{i} =0$ and the action of $J_{i}^{2}$ is given by multiplication by $j_{i}(j_{i}+1)$.
These relations mean that we can express any instance of $J_{n}$ say, by a summation of operators depending on $J_{i}$ for $i<n$. Thus a good basis of operator is for instance $J_{ij}$ for $i\neq j $ and $i,j<n$.
There are $(n-1)(n-2)/2$ such operators. They satisfy one relation that stems from the closure relation which is\be
\sum_{i\neq j <n} J_{ij} = j_{n}(j_{n}+1) - \sum_{i<n}j_{i}(j_{i}+1).
\ee
This makes it clear that if we want to maximally represent these operators we need $ n(n-3)/2$ labels.
These operators do not commute, therefore these labels represent an overcomplete basis.  A maximal commuting subalgebra  is of dimension $n-3$.

For example, in the case $n=4$ the basis is  labelled by $2$ integers while we need only one, and
for $n=5$  it is  labelled by $5$ integers while we need only two.  Despite this overcompleteness we will be able to determine all of the necessary properties of these states and we will discover some interesting relations between the orthogonal bases on the one hand and coherent intertwiners on the other.

%The fact that we need only one label in the $n=4$ case can be seen explicitly by constructing a 4-valent intertwiner by the composition of two 3-valent ones yeilding one virtual spin.  To this end we contract one spin from two copies of (\ref{C3}) to get
%\be \label{eqn_4S}
%  \bra S |z_i \ket = \Delta(j_1 j_2 S) \Delta(j_3 j_4 S) \sum_{m_1 m_2 m_3 m_4} i^{m_1 m_2 m_3 m_4}_{S} e^{j_1}_{m_1}(z_1) e^{j_2}_{m_2}(z_2) e^{j_3}_{m_3}(z_3) %e^{j_4}_{m_4}(z_4)
%\ee  
%where $S$ is the virtual spin associated to the composition and the coefficients are given by
%\be
%  i^{m_1 m_2 m_3 m_4}_{S} = \sqrt{2S+1} \sum_{m=-S}^{S} (-1)^{S-m} \threej{j_1}{j_2}{S}{m_1}{m_2}{m} \threej{S}{j_3}{j_4}{m}{m_3}{m_4}.
%\ee
%The states $|S\ket$ form an orthogonal basis of 4-valent intertwiners and are labelled by the one parameter $S$.  Notice that we made a choice to group the two spins as (12)(34).  If we had instead chosen to group the spins (13)(24) we would have a different basis which we call $|T\ket$.  We will see that the new basis (\ref{C}) which is labelled by the six $k_{ij}$ and hence two additional labels will be deeply related to the different orthogonal bases $|S\ket$ and $|T\ket$. 

\begin{figure} 
  \centering
    \includegraphics[width=0.8\textwidth]{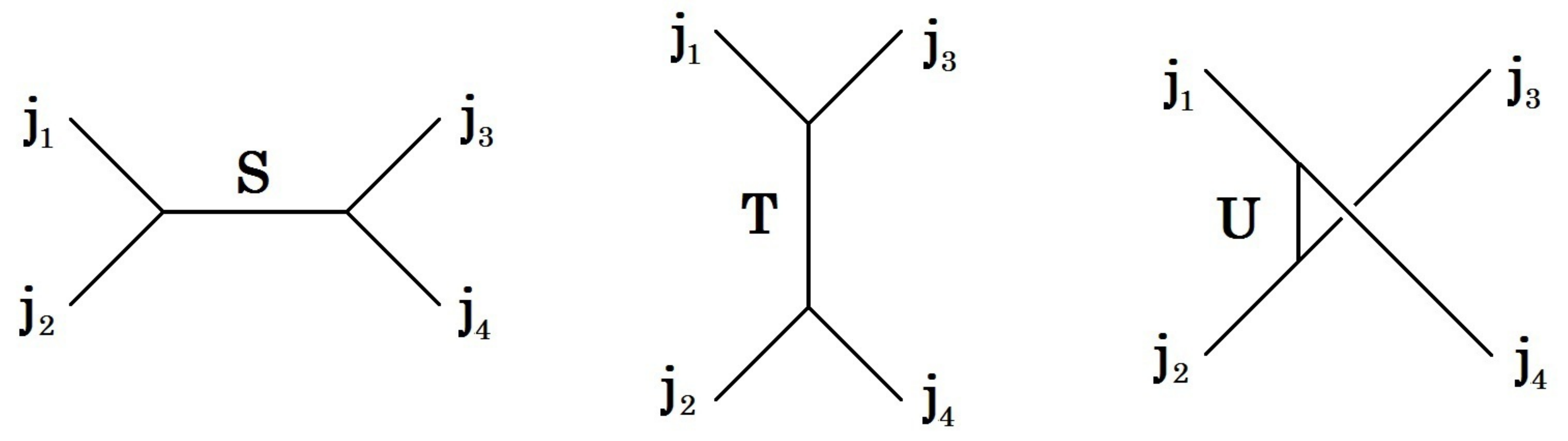}
    \caption{The three channels of a 4-valent vertex.}  \label{fig_STU}
\end{figure}

\section{The 4-valent case}
\label{4-val}

We now focus on the case $n=4$. A very convenient labelling of the basis $|{[k]}\ket$ is done in terms of three spins $S$, $T$, $U$ 
which refer to the three channels in which a 4-valent vertex can be split into two three valent ones.
The relationship between these labels and the $k$ labels is given by
\be\label{int1}
S\equiv j_{1}+j_{2} -k_{12},\quad T\equiv j_{1}+j_{3}  -k_{13},\quad U\equiv j_{1}+j_{4}-k_{14}.
\ee
where $S$, $T$, and $U$ are such that the $k_{ij}$ are non-negative integers.  The constraints in (\ref{kj}) imply that 
$ j_{1}+j_{2} -(j_{3}+j_{4}) = k_{12}-k_{34}$, thus we also have 
\be\label{int2}
S= j_{3}+j_{4}-k_{34},\quad T= j_{2}+j_{4}-k_{24},\quad U= j_{2}+j_{3}-k_{23}.
\ee
Summing over all $k_{ij}$ shows that $S$, $T$, and $U$ are not independent but satisfy the relation
 \be
 S+T+U=J.
 \ee 
We can therefore label the $4$-valent intertwiner basis by the four spins $j_i$ and two extra spins $S,T$ and we will henceforth denote by 
$k_{ij}(j_{i},S,T)$ the corresponding integers in (\ref{int1}, \ref{int2}). 
%Explicitely these are given by
%\bea
%k_{12}&=& j_{1}+j_{2}-S,\quad
%k_{13}= j_{1}+j_{3} -T,\quad
%k_{14}=  S +T - j_{2}-j_{3},\\
%k_{34}&=& j_{3}+j_{4}-S,\quad
%k_{24}= j_{2}+j_{4} -T,\quad
%k_{23}=  S +T - j_{1}-j_{4}.
%\eea
These integers cannot take arbitrary values, since  $k_{ij}$ are restricted by
$0\leq  k_{ij} \leq  \mathrm{max}(2j_{i},2j_{j})$, this restriction\footnote{It is given by
\bea
\mathrm{max}(|j_{1}-j_{2}|, |j_{3}-j_{4}|) \leq &S& \leq \mathrm{min}(j_{1}+j_{2}, j_{3}+j_{4}), \\
\mathrm{max}(|j_{1}-j_{3}|, |j_{3}-j_{4}|) \leq &T& \leq \mathrm{min}(j_{1}+j_{3}, j_{3}+j_{4}),\\
\mathrm{max}(j_{1}+j_{4}, j_{2}+j_{3}) \leq S &+&T \leq  J - \mathrm{max}(|j_{1}-j_{4}|, |j_{2}-j_{3}|).
\eea 
} is denoted by $(S,T)\in {\cal N}_{j_{i}}$.
In the case all spins are equal to $N/2$ this is simply $0\leq S,T\leq N$, $N\leq S+T\leq 2N$.

 We will  denote the corresponding basis by
  $|S,T\ket_{j_{i}}$ where
  % \footnote{Here we define $\bra C^{(4)}_{[k]} | z_i \ket \equiv C^{(4)}_{[k(j_i,S,T)]}(z_i)$.  Note that the states are orthogonal with respect to the $j_i$ labels by homogeneity.}
 \be \label{eqn_ST_notation}
 |S,T\ket_{j_{i}} \equiv  |[k](j_{i},S,T)\ket.
\ee
%for completeness we have to specify the choice of sign $\alpha$ in (\ref{C}). We chose  $\alpha=k_{24}$ which is dictated by having a symmetric expression of the completeness relation.
In the following we will omit the subscript $j_{i}$  and use the shorthand $|S,T\ket \equiv \left| S, T \right\ket_{j_{i}}$ for notational simplicity when the context is clear and the external spins are fixed.

\subsection{Overcompletness and Identity Decomposition}

As discussed above, the $|S,T\ket$ basis has one extra label and is thus overcomplete.  We will now investigate the nature of the relations among these states which is summarized by the following theorem:
\begin{theorem}
The $|S,T\ket$ states are  not linearly independent; all the relations among them are generated by the fundamental relation
\be \label{fund_rel}
 (k_{12}+1)(k_{34}+1) \left| S-1,T \right\ket - (k_{13}+1)(k_{24}+1) \left|  S,T-1 \right\ket +k_{14}k_{23} \left| S,T \right\ket =0 
\ee
where $k_{ij}$ stands for $k_{ij}(j_{i},S,T)$.
\end{theorem}

It turns out that the relation among the states is easily seen in the holomorphic representation.
It is well known that the gauge invariant quantities $[z_{i}|z_{j}\ket$ are not independent,
they satisfy the  Pl\"ucker relation:
\be \label{eqn_R_plucker}
R(z_i)\equiv [z_{1}|z_{2}\ket[z_{3}|z_{4}\ket - [z_{1}|z_{3}\ket[z_{2}|z_{4}\ket + [z_{1}|z_{4}\ket[z_{2}|z_{3}\ket =0.
\ee
In order to write the effect of this relation on the states $|S,T\ket_{j_{i}}$
lets compute first the effect of multiplication by one monomial
\bea\nn
[z_{1}|z_{2}\ket[z_{3}|z_{4}\ket ( z_{i} | S,T\ket_{j_{i}-\frac12}
% =[z_{1}|z_{2}\ket[z_{3}|z_{4}\ket  \prod_{i<j}\frac{ [z_{i}|z_{j}\ket^{k_{ij}(j_{i}-\frac12,S,T)}}{k_{ij}(j_{i}-\frac12,S,T)!}\\
%= (k_{12}k_{34})(j_{i},S,T )\prod_{i<j}\frac{ [z_{i}|z_{j}\ket^{k_{ij}(j_{i},S,T+1)}}{k_{ij}(j_{i},S,T+1)!}
= ({k_{12}k_{34}})(j_{i},S,T ) ( z_{i} | S,T+1\ket_{j_{i}}
\eea
where we used that $k_{12}(j_{i}-\frac12,S,T)  +1 = k_{12}(j_{i},S,T )= k_{12}(j_{i},S,T +1)$,
while 
$k_{13}(j_{i}-\frac12,S,T) = k_{13}(j_{i},S,T)-1 = k_{13}(j_{i},S,T+1)$,
and $k_{14}(j_{i}-\frac12,S,T) = k_{14}(j_{i},S,T)+ 1 = k_{14}(j_{i},S,T+1)$.
Performing similar computations for the different monomials we find that the multiplication by the Pl\"ucker relation can be implemented in terms of an operator $\hat{R} : {\cal{H}}_{j_{i}-\frac12} \rightarrow {\cal{H}}_{j_{i}}$
whose image vanishes identically.  It is defined by
 $  R(z_{i}) ( z_{i} | S,T\ket_{j_{i}-\frac12} = ( z_{i}| \hat{R} | S,T\ket_{j_{i}-\frac12}$ where $\hat{R}$ is given by 
\be
\hat{R} | S,T\ket_{j_{i}-\frac12} = 
{k_{12}k_{34}} | S,T+1\ket_{j_{i}} 
- {k_{13}k_{24}} |S+1,T\ket_{j_{i}} 
+ (k_{14}+2)(k_{23} +2)  |S+1,T+1\ket_{j_{i}} 
\ee
here $k_{ij}$ denotes $ k_{ij}(j_{i}, S, T)$.
By shifting the parameters $S\to S-1$ and $T\to T-1$ and using that 
$ k_{12}(j_{i}, S -1, T-1)= k_{12} +1$ etc. we obtain the desired relation stated in the theorem.
%The $|S,T\ket$ states are therefore not linearly independent; they satisfy relations which are generated by the Pl\"ucker relation.  That is $R(z_i) |S,T\ket = 0$ implies the following relation between the states
%\be \label{fund_rel}
% (k_{12}+1)(k_{34}+1) \left| S-1,T \right\ket - (k_{13}+1)(k_{24}+1) \left|  S,T-1 \right\ket +k_{14}k_{24} \left| S,T \right\ket =0 
%\ee
%where $k_{ij}$ means\footnote{Since $k_{ij}(S-1,T) = k_{ij}+1$ we can also write this relation as
% $$
% (k_{12}k_{34})(S-1,T) \left| S-1,T \right\ket - (k_{13}k_{24})(S,T-1) \left|  S,T-1 \right\ket +k_{14}k_{24}(S,T) \left| S,T \right\ket =0
% $$} $k_{ij}(j_{i},S,T)$.  
 By taking powers of the operator $\hat{R}$ we can generate many more relations which we will discuss in a later section.
Despite the linear dependence of these states they admit a resolution of identity, consistent with a coherent state basis:  
\begin{theorem} \label{thm_completeness}
The resolution of identity on the space of 4-valent intertwiners has the simple form
\be \label{eqn_res_id}
  \one_{{\cal H}_{j_{i}}} = \sum_{S,T}  \frac{|S,T\ket \bra S,T|}{\|S,T\|_{j_{i}}^2}, \qquad \|S,T\|^2_{j_{i}} \equiv \frac{(J+1)!}{\prod_{i<j}k_{ij}!}.
\ee
\end{theorem}

We give a proof of this theorem in Appendix \ref{app_completeness} which is specific to the 4-valent case.  The resolution of identity in the $n$-valent case has a similar form and follows from the relations (\ref{fundrelation}). % Furthermore, in \cite{FH_exact} it is shown that these states are related to the Livine-Speziale coherent intertwiner basis which is defined by $\bra j_i,z_i | = \int \rd g \, %\otimes_i ([ z_{i} |g )^{\otimes 2j_i}$ and introduced in \cite{coh1}.  The overlap between these bases is given by
%\be \label{eqn_coherent_st_id}
%  \frac{\bra   j_i, z_i \|S,T \ket}{\sqrt{\prod_i (2j_i)!}} =  \bra  z_i |S,T \ket.
%\ee  
We will show that despite the fact that they are discrete, the $|S,T\ket$ basis shares many of the same properties as the coherent intertwiners such as the correspondence with classical tetrahedra in the semi-classical limit. In addition the $|S,T\ket$ states also possess a simple relation with the orthogonal basis as  we will show in the next section .

\subsection{The Relation with the Orthogonal Basis}

In the previous sections we introduced a new and overcomplete basis of the space of 4-valent intertwiners which provided a simple decomposition of the identity.  On the other hand, the standard basis of 4-valent intertwiners is orthogonal, and is defined by the eigenstates of either of the invariant operators $J_{1}\cdot J_{2}$ or $J_{1}\cdot J_{3}$ or $J_{1}\cdot J_{4}$.  We will denote these orthogonal bases by $|S\ket$ and $|T\ket$ and $|U\ket$ respectively.  We would now like to investigate the action of the $S$ and $T$ channel operators $J_{1}\cdot J_{2}$ and $J_{1}\cdot J_{3}$ on $\left|S,T\right\ket$ as well as the relationship between the four  bases: $\left|S,T\right\ket$, $\left|S\right\ket$, $\left|T\right\ket$, $ | U\ket$. 

It is well known that, up to normalization,
the usual 4-valent intertwiner basis is obtained by the composition of two trivalent intertwiners.
For now we will focus on the $|S\ket$ states, which in the holomorphic representation, are defined to be
\bea \label{eqn_S_def}
\left( z_i | S \right\ket \equiv \int \rd\mu(z) {C}_{(j_{1},j_{2},S)}(z_{1},z_{2},\check{z}) {C}_{(S,j_{3},j_{4})}(z,z_{3},z_{4}),
%(z_{1},z_{2},\check{z}|k_{ij}(j_{1},j_{2},S)\ket (z,z_{3},z_{4}|k_{ij}(S,j_{3},j_{4})\ket. %\\ \label{eqn_T_def}
%\left\bra T \right.\left| z_i \right\ket \equiv \int \rd\mu(z) C_{(j_{1},j_{3},T)}(z_{1},z_{3},\check{z}) C_{(T,j_{2},j_{4})}(z,z_{2},z_{4}).
\eea
where $|\check{z}\ket \equiv |z]$ and ${C}_{(j_{1},j_{2},S)}(z_{1},z_{2},\check{z}) = (z_{1},z_{2},\check{z}|k_{ij}(j_{1},j_{2},S)\ket$.  As shown in \cite{OH, LFetera1} the operators $J_{i}\cdot J_{j}$ in the holomorphic representation can be  written in terms of the SU$(N)$ operators
\be
E_{ij}\equiv z_{i}^{A}\partial_{z_{j}^{A}},
\ee 
as the quadratic combinations
\be
  2J_{i}\cdot J_{j} = E_{ij}E_{ji}-\frac12 E_{ii}E_{jj} - E_{ii}.
\ee

The operator $E_{ij}$ acts nontrivially only on a function of $z_{j}$ and its action amounts to replacing $z_{j}$ by $z_{i}$, i-e
 \be
E_{ij}\cdot [z_{j}|w\ket = [z_{i}|w\ket.
 \ee
Using this we can now compute the action of $J_1 \cdot J_2$ on $|S\ket$.  First note that the action of $E_{ii}$ on $|S\ket$ is given by $2j_{i}$ and the action of $E_{12}E_{21}$ is given by $(j_1-j_2+S)(-j_1+j_2+S+1)$.  Therefore the action of $J_1 \cdot J_2$ on $|S\ket$ is found to be
\bea \label{eqn_eigen_J_dot_J}
  J_{1}\cdot J_{2} \left|S\right\ket &=& \frac12\left(S(S+1) - j_{1}(j_{1}+1) - j_{2}(j_{2}+1)\right) \left|S\right\ket.% \\
%  J_{1}\cdot J_{3} \left|T\right\ket &=& \frac12\left(T(T+1) - j_{1}(j_{1}+1) - j_{3}(j_{3}+1)\right) \left|T\right\ket.
\eea

We are now in a position to discuss the physical interpretation of the spins $S$ and $T$.  From equation (\ref{eqn_eigen_J_dot_J}) we see that the operator $(J_1 +J_2)^2$ is diagonal in the $|S\ket$ basis with eigenvalue $S(S+1)$.  In \cite{Baez:1999tk} it is pointed out that if $A_1$ and $A_2$ are the classical area vectors of two faces of a tetrahedron then $|A_1 + A_2|^2$ is equal to four times the area of the medial parallelogram between the two faces.  The spins $T$ and $U$ would then be the areas of the other two medial parallelograms in the tetrahedron.  

This interpretation, however, does not hold for the $|S,T\ket$ states as we will see by computing the action $J_1 \cdot J_2$ on $|S,T\ket$.  We will find the true correspondence with the classical variables when we study the semi-classical limit.

\begin{theorem}
The action of $J_1 \cdot J_2$ on $|S,T\ket$ does not change the value of $S$ and it is given by
\bea \label{Jact1}
2J_{1}\cdot J_{2} \left|S,T \right\ket = \left(S(S+1) - j_{1}(j_{1}+1) - j_{2}(j_{2}+1)\right) \left|S,T \right\ket \\
+\left( (k_{14}+1)(k_{23}+1)  \left|S,T-1 \right\ket -
k_{14}k_{23}  \left|S,T \right\ket  \right)  \nonumber \\
+\left( (k_{13}+1)(k_{24}+1) \left|S,T+1 \right\ket -
k_{13}k_{24}  \left|S,T \right\ket  \right). \nonumber
\eea
where $k_{ij}$ stands for $k_{ij}(j_i,S,T)$.  Similarly the action of $J_1 \cdot J_3$ does not change the value of $T$.
\end{theorem}
\begin{proof}
The action of $E_{ii}$ on $|S,T\ket$ is given by $2j_{i}$ while the action of $E_{12}E_{21}$ on $|S,T\ket$ is
\bea
  \left( k_{13}(k_{23}+1) + k_{14}(k_{24}+1) \right)|S,T\ket + (k_{13}+1)(k_{24}+1)|S,T+1\ket + (k_{14}+1)(k_{23}+1)|S,T-1\ket. \nonumber
\eea
Now with this and the relation 
\be
  k_{13}k_{23} + k_{14}k_{24} = S^{2}- (j_{1}-j_{2})^{2} - k_{13}k_{24}-k_{14}k_{23}
\ee
we find the desired result.  The action of $J_1 \cdot J_3$ can be deduced from a permutation exchanging $1$ and $3$, under such a permutation $J_1 \cdot J_3 \to J_1 \cdot J_2$ and $ (-1)^{k_{23}} |S,T\ket \to  |T, S\ket$ .
Similarly under an exchange of $1$ and $4$, $J_1 \cdot J_4 \to J_1 \cdot J_2$ and $ (-1)^{k_{23}+k_{34}} |S,T\ket \to  |U, T\ket$ .
\end{proof}
%Similarly, one can compute the action of $J_1 \cdot J_3$
%\bea \label{Jact2}
%2J_{1}\cdot J_{3} \left|S,T \right\ket = \left(T(T+1) - j_{1}(j_{1}+1) - j_{3}(j_{3}+1)\right) \left|S,T \right\ket \\
%+\left( (k_{12}+1)(k_{34}+1)  \left|S+1,T \right\ket -
%k_{12}k_{34}  \left|S,T \right\ket  \right) \nonumber \\
%+\left( (k_{14}+1)(k_{23}+1) \left|S-1,T \right\ket -
%k_{14}k_{23}  \left|S,T \right\ket  \right). \nonumber
%\eea

While the $S$ and $T$ spins don't share the interpretation of areas of parallelograms like in the orthogonal basis (since there are extra diagonal terms), it turns out that they are still closely related as we will now show.  First of all, notice that the coefficient of the first term in (\ref{Jact1}) is the same as the eigenvalue in (\ref{eqn_eigen_J_dot_J}).  Furthermore, if one sums over $T$ in (\ref{Jact1}) it can be seen that the last two terms cancel out because $k_{13}(j_{i},S,T-1)= k_{13}(j_{i},S,T)+1$,
$k_{14}(j_{i},S,T+1)= k_{14}(j_{i},S,T)+1$... and so on.  
Therefore $\sum_{T} \left|S,T \right\ket$ is {\it  proportional } to $\left| S \right\ket$.  What we will now show in the following theorem is that the proportionality constant is exactly one.
\begin{theorem} \label{thm_sum_T}
The orthogonal basis is obtained from the $\left|S,T \right\ket$ basis by summing over the $S$ or $T$ channels
\be
\left| S \right\ket =\sum_{T}  \left|S,T \right\ket, \hspace{12pt} \left| T \right\ket =\sum_{S} (-1)^{k_{23}}  \left| S,T \right\ket, \quad \left| U \right\ket =\sum_{S+T=J-U} (-1)^{k_{23}+k_{34}}  \left| S,T \right\ket.
\ee
\end{theorem}
\begin{proof}
Using the generating functionals in (\ref{defC}) in analogy with the definition (\ref{eqn_S_def}) of $|S\ket$ we can perform the following Gaussian integral 
%\be
%  {\cal C} (\tau_{1},\tau_{2}, \tau_{12};z_{1},z_{2},\check{z}) = \sum_{k_1,k_2,k_{12}} \tau_{1}^{k_{1}} \tau_{2}^{k_{2}} \tau_{12}^{k_{12}} %\frac{[z_1|\check{z}\ket^{k_1}}{k_{1}!} \frac{[z_2|\check{z}\ket^{k_2}}{k_{2}!} \frac{[z_1|z_2\ket^{k_{12}}}{k_{12}!} 
%\ee
%and
%\be
%  {\cal C} (\tau_{34},\tau_{3},\tau_{4};z,z_{3},z_{4}) = \sum_{k_3,k_4,k_{34}} \tau_{34}^{k_{34}} \tau_{3}^{k_{3}} \tau_{4}^{k_{4}} \frac{[z_3|z_4\ket^{k_{34}}}{k_{34}!} %\frac{[z_3|z\ket^{k_3}}{k_{3}!} \frac{[z_4|z\ket^{k_{4}}}{k_{4}!}
%\ee
%\be \label{eqn_C_gen_fun}
%  {\cal C}_{\tau_{ij}}(z_{i}) \equiv \sum_{[k]} \prod_{i<j} \tau_{ij}^{k_{ij}} C^{(n)}_{[k]}(z) = \sum_{[k]} \prod_{i<j} \tau_{ij}^{k_{ij}} \frac{[z_i|z_j\ket^{k_{ij}}}{k_{ij}!} = e^{ \sum_{i<j} \tau_{ij} [z_i|z_j\ket }.
%\ee
\bea \label{eqn_gen_fun_S_ST}
  && \int \rd\mu(z){\cal C}_{(\tau_{1},\tau_{2},\tau_{12})}(z_{1},z_{2},\check{z}) {\cal C}_{(\tau_{3},\tau_{4},\tau_{34})}(z,z_{3},z_{4})  \\
  &=& e^{\tau_{12}[z_1|z_2\ket + \tau_{34}[z_3|z_4\ket} \int \rd\mu(z) e^{\tau_{1}[\check{z}|z_1\ket + \tau_{2}[\check{z}|z_2\ket} e^{\tau_{3}[z|z_3\ket + \tau_{4}[z|z_4\ket} 
  = e^{\sum_{i<j} \tau_{ij} [z_i|z_j\ket} 
  = {\cal C}_{(\tau_{ij})}(z_{i}), \nonumber
\eea
%where
%\be
%  {\cal C}(\tau_{ij};z_{1},z_{2},z_{3},z_{4}) = \sum_{[k]} \prod_{i<j} \tau_{ij}^{k_{ij}} \frac{[z_i|z_j\ket^{k_{ij}}}{k_{ij}!}.
%\ee
%From these definitions we can compute the integral (\ref{eqn_S_def}) by check that the intetwinner $\cal C$ behaves in a self similar manner under the  composition, that is by performing a gaussian integration we get
where $|\check z\ket = |z]$ and $\tau_{13}=\tau_{1}\tau_{3}$, $\tau_{14}=\tau_{1}\tau_{4}$, $\tau_{23}=\tau_{2}\tau_{3}$, $\tau_{24}=\tau_{2}\tau_{4}$.  Now let $k_{1} = j_{1}-j_{2}+S$, $k_{2} = j_{2}-j_{1}+S$, $k_{3} = j_{3}-j_{4}+S$, and $k_{4} = j_{4}-j_{3}+S$ as prescribed by (\ref{eqn_3_k}). Then looking at the coefficient of 
\be
\tau_{12}^{k_{12}}\tau_{1}^{k_{1}} \tau_{2}^{k_{2}}\tau_{3}^{k_{3}} \tau_{4}^{k_{4}}\tau_{34}^{k_{34}}
\ee
we get the conditions $k_{1} = k_{13}+k_{14}$, $k_{2} = k_{23}+k_{24}$, $k_{3} = k_{13}+k_{34}$, and $k_{4} = k_{14}+k_{24}$.  These conditions are trivially satisfied if the $k_{ij}$ are defined as in (\ref{int1}) and (\ref{int2}) which can be seen for instance by adding $k_{12}$ to the first condition.  Notice, however that the LHS of (\ref{eqn_gen_fun_S_ST}), when expanded, is a sum over $j_i$ and $S$ whereas the RHS is a sum over $j_i$, $S$, and $T$.  Thus we get the identity
\be
\int \rd\mu(z) {C}_{(j_{1},j_{2},S)}(z_{1},z_{2},\check{z}) {C}_{(S,j_{3},j_{4})}(z,z_{3},z_{4})
= \sum_{T} ( z_{i}| S,T\ket_{j_{i}}.
\ee
which implies $|S\ket = \sum_T |S,T\ket$.The other identities are obtained by permutation of indices.
% and a similar calculation shows that $|T\ket = \sum_S (-1)^{k_{23}}|S,T\ket$.  %Now to show $|T\ket = \sum_S |S,T\ket$ consider the following Gaussian integration
%We can perform the following Gaussian integral
%\be
%\int \rd\mu(z){\cal C}_{(\tau_{1},\tau_{3},\tau_{13})}(z_{1},z_{3},\check{z}) {\cal C}_{(\tau_{24},\tau_{4},\tau_{2})}(z,z_{4},z_{2}) = {\cal %C}_{(\tau_{ij})}(z_{1},z_{2},z_{3},z_{4}).
%\ee
%where $\tau_{12}=\tau_{1}\tau_{2}$, $\tau_{14}=\tau_{1}\tau_{4}$, $\tau_{23}=\tau_{2}\tau_{3}$ and $\tau_{34}=\tau_{3}\tau_{4}$.  Looking at the coefficient of 
%$$
%\tau_{13}^{j_{1}+j_{3}-T}\tau_{1}^{j_{1}-j_{3}+T} \tau_{3}^{j_{3}-j_{1}+ T}\tau_{2}^{j_{2}-j_{4}+T} \tau_{4}^{j_{4}-j_{2}+T}\tau_{24}^{j_{2}+j_{4}-T}
%$$
%we get the identity
%\be
%\int \rd\mu(z) {C}_{(j_{1},j_{3},S)}(z_{1},z_{3},\check{z}) {C}_{S,j_{2},j_{4}}(z,z_{2},z_{4})
%= \sum_{S} C_{j_{i}}^{S,T}(z_{i}).
%\ee
%Hence $|T\ket = \sum_S |S,T\ket$.  
\end{proof}
\begin{figure} 
  \centering
    \includegraphics[width=0.7\textwidth]{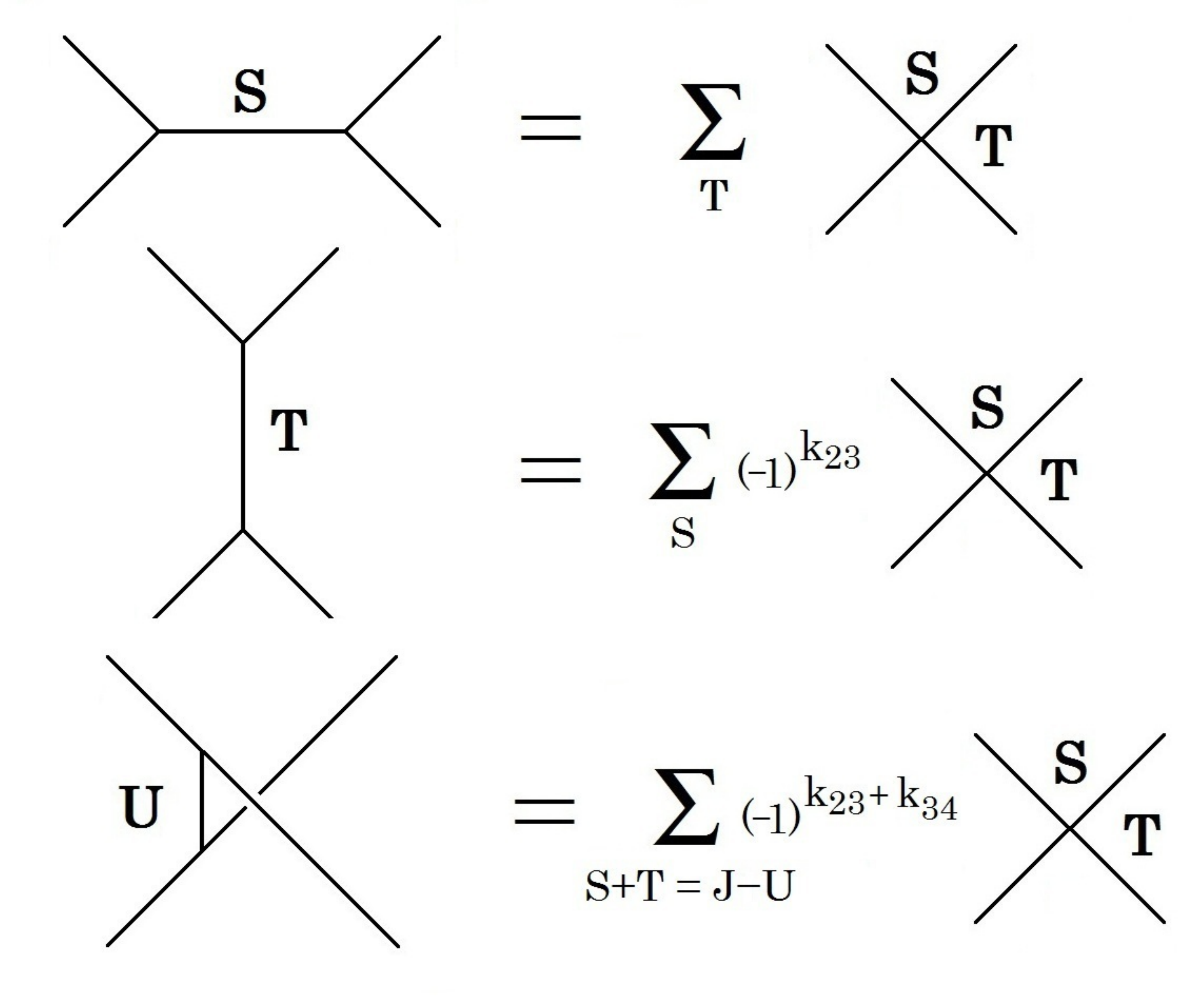}
    \caption{The graphical representation of theorem \ref{thm_sum_T} where a 4-valent vertex labelled by S and T is summed over T to produce the S channel decomposition.}  \label{fig_sumT}
\end{figure}
This last theorem shows that the $|S\ket$ and $|T\ket$ or $|U\ket$ bases are generated by the $|S,T\ket$ basis.  This is particularly useful for instance when describing spin-network amplitudes containing 4-valent nodes since a choice of $S$ or $T$ basis must be made at every such node.  The amplitude written in the $|S,T\ket$ basis however will generate all the different kinds of amplitudes by simply summing over the labellings.  For example the $15j$ symbol comes in five different kinds depending on the basis choice at the five nodes.  Thus a new symbol labelled by 20 spins, i.e. ten $j_i$, five $S$'s and five $T$'s, based on the $|S,T\ket$ basis would be a generator of these various symbols.  Moreover this $20j$ symbol would be the amplitude corresponding to the coherent 4-simplex.  We will define this new symbol shortly.

\section{Scalar Products}

In this section we will compute the scalar product in the $|S,T\ket$ basis and demonstrate the utility of Theorem \ref{thm_sum_T} by generating all the various other scalar products.  Let us first make a general remark about the form of the scalar product that follows from the  resolution of identity in (\ref{eqn_res_id}).
Let us split the scalar product into the naive product and the remainder:
\be
\left\bra S,T  \right.\left| S',T'  \right \ket = \|S,T\|^2 \delta_{S,S'} \delta_{T,T'} + O_{S,T}^{S',T'}
\ee
The resolution of the identity implies that 
\be
\sum_{S',T'} O_{S,T}^{S',T'} \frac{\bra S',T' |}{||S',T'||^{2}} =0= \sum_{S,T}\frac{|S,T\ket}{||S,T||^{2}} O_{S,T}^{S',T'}
\ee
This means that the reminder belongs to the algebra generated by the  fundamental relation in (\ref{fund_rel}).
These relations can be derived by considering the product of the operator $\hat{R}:H_{j_{i}-\frac12} \mapsto {\cal H}_{j_{i}}$ introduced previously: $R(z_i)^N \bra z_{i} |S,T\ket = \bra z_{i} | \hat{R}^{N}|S,T\ket$, where
 $R(z_i)$ is the Pl\"ucker relation given in (\ref{eqn_R_plucker}).
 Expanding $R(z_i)^N$ using the multinomial theorem we find
 \be\label{rel}
\left(R(z_i)\right)^N \prod_{i<j} [z_i|z_j\ket^{k_{ij}(j_i-N/2,s,t)} = \sum_{S,T} R^{(s,t)}_{(S,T)}(N) \prod_{i<j} [z_i|z_j\ket^{k_{ij}(j_i,S,T)} =0,%\quad \mathrm{if}\quad N \in \Z^+.
 \ee
where the summation coefficients are given by
\be\label{eqn_R_def}
R^{(s,t)}_{(S,T)}(N) = \frac{(-1)^{t-T+N}N!}{[s-S+N]![t-T+N]![S-s+T-t-N]!}
\ee
and the sum is over $S=s+N-a$, $T=t+N-b$ with $ a,b\geq 0$ and $ a+b \leq N$.  
From this result and the definition of the states $ \bra z_{i} |S,T\ket = ||S,T||_{j_{i}} / (J+1)!  \prod_{i<j} [z_i|z_j\ket^{k_{ij}(j_i,S,T)}$, we can write this relation  as
\be
 \frac{ \hat{R}^{N}  | s,t\ket_{j_{i}-N/2}}{ ||s,t||_{j_{i}-N/2}^{2 }}  =  \frac{(J+1)!}{(J-2N +1)!} 
 \left(\sum_{S,T} R^{(s,t)}_{(S,T)}(N) \frac{| S,T \ket_{j_{i}}}{ ||S,T ||_{j_{i}}^{2 }}\right) =0.
\ee
The coefficients in the sum vanish if any of the arguments in the factorials is negative.  Note that for $N=1$ we recover the fundamental relation (\ref{fund_rel}).  

Now that we have determined the linear relations among the basis states we can deduce the exact form of the scalar product
\begin{lemma}\label{product1}
The scalar product is given by
\be\label{scalar}
\left\bra S,T  \right.\left| S',T'  \right \ket = \|S,T\|^2 \delta_{S,S'} \delta_{T,T'} +
\sum_{s,t,N} \frac{(-1)^N }{N!}\frac{(J-N+1)!}{\prod_{i<j} k_{ij}(j_i-N/2,s,t)!} 
{ R^{(s,t)}_{(S,T)}(N) R^{(s,t)}_{(S',T')}(N) }{}
\ee
\end{lemma}
The proof of this formula is given  in appendix \ref{alpha_proof}.

\subsection{Constraints Quantisation}

We would like now to develop a deeper understanding of the construction just given of the scalar product.  We have seen that the complexity of the scalar product comes from the imposition of the constraints $\hat{R}=0$. This suggest that we should be able to understand the previous construction in terms of constraint quantization.
In order to do so, lets introduce the auxiliary  Hilbert space ${\cal \widehat{H}}_{j_{i}}$ with an orthogonal basis 
$|S,T)_{j_{i}}$ having $(S,T)\in {\cal N}_{j_{i}}$ and the scalar product
\be
(S',T'|S,T) =||S,T||_{j_{i}}^{2} \delta_{S,S'}\delta_{T,T'}.
\ee 
For this Hilbert space the decomposition of the identity takes the canonical form
\be
\one_{{\cal \widehat{H}}_{j_{i}}} = \sum_{S,T}  \frac{|S,T)( S,T|}{\|S,T\|_{j_{i}}^2}.
\ee
We  define the operator $\hat{R} : {\cal \widehat{H}}_{j_{i}-\frac12}\mapsto {\cal \widehat{H}}_{j_{i}}$ by
\be
\hat{R}|S,T)_{j_{i}-\frac12} \equiv 
{k_{12}k_{34}} | S,T+1)_{j_{i}} 
- {k_{13}k_{24}} |S+1,T)_{j_{i}} 
+ (k_{14}+2)(k_{23} +2)  |S+1,T+1)_{j_{i}} 
\ee
Its powers can be evaluated in terms of the coefficients introduced it the previous section, we find 
\be\label{matrixel}
{}_{j_{i}} (S,T| \hat{R}^{N}  | s,t)_{j_{i}-N/2} = ||s,t||_{j_{i}-N/2}^{2 }  \frac{(J+1)!}{(J-2N +1)!} 
  R^{(s,t)}_{(S,T)}(N).
\ee
The operator $\hat{R}$ is not hermitian, however the operator 
$$H \equiv \hat{R}^{\dagger} R $$ is an hermitian operator, being positive its kernel coincides with the kernel of $\hat{R}$. 
The intertwiner Hilbert space is defined as the quotient of this auxiliary Hilbert space by the relation $H=0$. This means that $ {\cal {H}}_{j_{i}} = \mathrm{Im}\Pi_{j_{i}} $, where $\Pi^{2}_{j_{i}}=\Pi_{j_{i}}$ with $\Pi_{j_{i}} : {\cal \widehat{H}}_{j_{i}}\to {\cal \widehat{H}}_{j_{i}}$ the projector onto the 
kernel of $H$.  This means that the intertwined states are related to the auxiliary states as 
$$ |S,T\ket_{j_{i}} = \Pi_{j_{i}} |S,T)_{j_{i}}$$ and the physical scalar is given by the matrix element of the projector
\be
\left\bra S,T  \right.\left| S',T'  \right \ket = ( S,T | \Pi_{j_{i}} |  S',T' ).
\ee
From the results of the previous section this projector can be explicitly constructed.
\begin{lemma}\label{product2}
The projector onto the kernel of $H$ is explicitly given by
\be
\Pi_{j_{i}} = 1 +\sum_{N=1}^{\mathrm{min}(2j_{i})}\frac{(-1)^{N}}{N!} \frac{(J-N+1)!(J-2N+1)!}{(J+1)!^{2}} \, \hat{R}^{N}(\hat{R}^{\dagger})^{N}.
\ee
\end{lemma}
The proof is given in  appendix \ref{alpha_proof}.

\subsection{Overlap with the Orthogonal Basis}

Let us now show how theorem \ref{thm_sum_T} can be used to generate the various other scalar products.  In appendix \ref{app_R_delta} we show  that we have the following identity
\be \label{eqn_R_delta}
  \sum_{T} R^{(s,t)}_{(S,T)}(N) = \delta_{s,S}. %\hspace{12pt} \sum_{S} (-1)^{k_{23}(S,T)} R^{(s,t)}_{(S,T)}(N) = \delta_{t,T} (-1)^{k_{23}(s,t)}.
\ee
This identity translates into the statement that $\hat{R}^{\dagger}$ acts diagonally on $|S)_{j_{i}}=\sum_{T}|S,T)_{j_{i}}$:
\be
(R^{\dagger})^{N} |S)_{j_{i}} = \frac{(J-2N +1)!}{(J+1)!} |S)_{j_{i}-N/2}.
\ee
For details on proving this identity see appendix \ref{app_R_delta}.  Therefore summing over $T$ in (\ref{scalar}) yields
\bea \label{eqn_s_ST_overlap}
  \left\bra S  \right.\left|S',T' \right \ket &=&
  \sum_{N=0}^{\mathrm{min}(2j_{i})}  \alpha_{J,N} \sum_{t} R^{(S,t)}_{(S',T')}(N) ||S,t||_{j_{i}-N/2}^{2}
\eea
where it is convenient to define
\bea \label{eqn_alpha}
 \alpha_{J,N} &\equiv & 
 \frac{(-1)^N}{N!} \frac{(J-N+1)!}{(J-2N+1)!}.
\eea
By summing over the different labels in (\ref{scalar}) we can compute the remaining scalar products:
\bea \label{eqn_usual_sc_prods}
  \left\bra S\right|\left.S' \right\ket = \delta_{S,S'}\sum_{T,N} \alpha_{J,N} ||S,T||_{j_{i}-N/2}^{2}, \hspace{12pt}
   \left\bra T\right|\left.T' \right\ket = \delta_{T,T'} \sum_{S,N} \alpha_{J,N} ||S,T||_{j_{i}-N/2}^{2},\nonumber \\ \hspace{12pt}  \bra S | T \ket_{j_{i}} = (-1)^{k_{23}} \sum_N \alpha_{J,N} ||S,T||_{j_{i}-N/2}^{2} 
   =  \sum_N \alpha_{J,N} (S|T)_{j_{i}-N/2}
\eea
where here the sum over $N$ starts from zero.  Hence all of the scalar products between $|S\ket$, $|T\ket$, and $|S,T\ket$ bases are different summations over $\alpha_{J,N}$ and the canonical norms.  %We will now go on to compute these coefficients, the proof of which is contained in appendix \ref{alpha_proof}.
%with
%\be
%{\cal{N}}(j_{i},s,t) \equiv \sqrt{(j_{1}+j_{2}-s)! (j_{3}+j_{4}-s)! (j_{1}+j_{3}-t)!(j_{2}+j_{4}-t)! (j_{1}+j_{4}-u)!(j_{2}+j_{3}-u)!},
%\ee
%with$s+t+u = \sum_{i}j_{i}$.
%In other words ${\cal{N}}(j_{i},s,t)=\sqrt{\prod_{i<j} (k_{ij})!}$ where $k_{ij}$ are determined by $j_{i}$ and $S,T$.
%
In appendix \ref{sec_scalar_products} we show how to perform the summations in (\ref{eqn_usual_sc_prods}) to give the well known normalization factors for $|S\ket$ and $|T\ket$.  They are given by
\be
  \left\bra S\right|\left.S' \right\ket  = \frac{\delta_{S,S'}}{2S+1}  \Delta^{2}(j_{1}j_{2}S) \Delta^{2}(j_{3}j_{4}S), \hspace{12pt} \left\bra T\right|\left.T' \right\ket  = \frac{\delta_{T,T'}}{2T+1}  \Delta^{2}(j_{1}j_{3}T) \Delta^{2}(j_{2}j_{4}T),
\ee
where the triangle coefficients were given in (\ref{eqn_tri_coeff}).  Finally, it is easy to see that the overlap between the $|S\ket$ and $|T\ket$ bases is given by a 6j symbol.  That is, the third sum in (\ref{eqn_usual_sc_prods}) can be recognized as the Racah expansion of the 6j symbol by making a change of variable $m = J-N$.  Doing so we get
\bea
  \bra S | T \ket 
  &=& {(-1)^{J+k_{23}}}{\Delta(j_{1}j_{2}S)\Delta(j_{3}j_{4}S)\Delta(j_{1}j_{3}T)\Delta(j_{2}j_{4}T)} \sixj{j_1}{j_2}{S}{j_4}{j_3}{T}.
\eea

These relations with the orthogonal basis provide a consistency check, but they will also be useful later in connecting with the 15j symbol.  To do so we will next study the contraction of the $|S,T\ket$ states.  %We will now go on to describe the semi-classical behaviour of the $|S,T\ket$ states by studying the asymptotics of the scalar product. 

%\begin{figure} 
%  \centering
%    \includegraphics[width=0.8\textwidth]{scalar}
%    \caption{The graphical representation of the different scalar products.}  \label{fig_scalar}
%\end{figure}

\section{Spin Network Amplitudes}

In this section we show how the discrete coherent basis (\ref{C}) is a natural basis for representing and computing spin network amplitudes.  A spin network can be defined by a directed graph $\Gamma$ in which the edges are labelled by spins $j_e$
and vertices are labelled by intertwiners. The spin network amplitude is obtained by contracting the intertwiners along the  edges of $\Gamma$. 
Depending on the intertwiner basis we get different amplitudes.  The two bases we consider here are the continuous basis $\|j_{i},z_{i}\ket $ and the discrete basis $|[k_{ij}]\ket$. 

More precisely lets denote by $\{\Omega_v\}_{v\in V_{\Gamma}} $ a choice of intertwiner for every vertex of $\Gamma$ where $\|\Omega_v\ket \in \mathrm{Inv}_{\mathrm{SU(2)}}(\bigotimes_{e\supset v}V^{j_{e}})$ is represented by a holomorphic function $( w_e | \Omega_v \ket$.
The contraction of the intertwiners  is accomplished using the scalar product (\ref{barg_in_prod}) and is denoted
\be
  \underset{v}{\corner} \|\Omega_{v} \ket \equiv \int \prod_{e \in E_\Gamma} \rd\mu(w_e) \prod_{v \in V_\Gamma} ( w_e \| \Omega_v \ket,
\ee
where $E_\Gamma$ and $V_\Gamma$ are the set of edges and vertices of $\Gamma$.  Note that in the contraction $w_{e^{-1}} \equiv \check{w}_e$.  In particular, the contraction of coherent intertwiners produces an amplitude which is given by
\be \label{amplitude}
  A_\Gamma(j_e,z_e) \equiv \underset{v\in V_\Gamma}{\corner} {\|j_e,z_e \ket} = \int \prod_{v\in V_\Gamma} \rd g_v \prod_{e \in E_\Gamma} \frac{ [z_e|g_{s_e}g_{t_e}^{-1}|z_{e^{-1}}\ket^{2j_e}}{(2j_{e})!}.
\ee

The amplitude for a general graph, in the discrete basis, will depend on the integers $k_{ee'}^{v}$ associated with each pair of edges meeting at $v$ and is denoted 
$$A_\Gamma(k_{ee'}^{v})\equiv \underset{v\in V_\Gamma}{\corner} |k^{v}_{ee'}\ket.$$  
The fundamental relation (\ref{fundrelation}) between the two bases implies that these two amplitudes are related as follows
\bea
  A_\Gamma(j_e,z_e) &=& \sum_{k^{v}_{ee'}\in K_j} A_{\Gamma}(k_{ee'}^{v}) \prod_v \frac{(z_{e}| k^{v}_{ee'}\ket}{\|[k^{v}]\|^2}\nn \\
  &=&  \sum_{k^{v}_{ee'}\in K_j} A_{\Gamma}(k_{ee'}^{v})
\prod_{v}  \left(\frac{\prod_{(ee')\supset v} [z_{e}|z_{e'}\ket^{k_{ee'}^{v}}}{(J_{v}+1)!}\right). \label{eqn_discrete_amp}
\eea

 %then this then implies that $\bra z_i | [k_{ij}]\ket = \bra j_i,z_i \| [k_{ij}] \ket$.  Hence the functions (\ref{C}) are simply the overlap of these two bases.

%As shown in \cite{FH_exact} the overlap between these two basis is given by
%\be
%\frac{\|j_{i},z_{i}\ket }{\sqrt{\prod_{i } 2j_{i}!}} = \sum_{[k_{ij}]} \frac{\overline{\bra z_{i} |[k_{ij}]\ket}  | [k_{ij}] \ket}{ ||[k_{ij}]||^{2}} ,\qquad ||[k_{ij}]||^{2} \equiv \frac{(J +1)!}{\prod_{i<j} k_{ij}!}.
%\ee

%The amplitude in the spinor basis, for a general graph $\Gamma$, depends on one pair of spinors $|z_e\ket \neq |z_{e^{-1}}\ket$ per edge.  The spinors are defined for the two directions of the edge and they label a basis of intertwiners at each node.  
%The contraction of these intertwiners defines an amplitude which is given by
%\be \label{amplitude}
%  A_\Gamma(j_e,z_e) \equiv \underset{v\in V_\Gamma}{\corner} \frac{\|j_e,z_e \ket}{\sqrt{(2j_e)!}} = \int \prod_{v\in V_\Gamma} \rd g_v \prod_{e \in E_\Gamma} [z_e|g_{s_e}g_{t_e}^{-1}|z_{e^{-1}}\ket^{2j_e}
%\ee 
%where $\rd g$ is the normalized Haar measure on SU(2), $V_\Gamma$ is the set of vertices of $\Gamma$, $E_\Gamma$ is the set of edges of $\Gamma$, and $s_e$, $t_e$ label the starting, terminal vertices of the edge $e$.
We would now like to evaluate these  amplitudes.  In \cite{FH_exact} it is shown how to evaluate the integrals in (\ref{amplitude}) by comparing the coefficients of the same homogeneity in the following generating functional 
\be \label{eqn_gen_func}
{\cal A}_{\Gamma}(z_{e}) \equiv \sum_{j_e} \int \prod_{v\in V_\Gamma} \rd g_v (J_v + 1)! \prod_{e \in E_\Gamma} \frac{[z_e|g_{s_e}g_{t_e}^{-1}|z_{e^{-1}}\ket^{2j_e}}{(2j_e)!} = \frac{1}{(1+\sum_{C} A_C(z_e))^2}
\ee 
where $J_v$ is the sum of the spins at the vertex $v$.  

The sum is over collections $C = \{c_1,...,c_k\}$ of non-trivial cycles of the graph which are disjoint, i.e. do not share any edges or vertices with themselves or the other cycles.  The quantities $A_C \equiv A_{c_1} \cdots A_{c_k}$ are defined for each cycle $c_i = (e_1,...,e_n)$ by
\be
  A_{c_i}(z_e) \equiv -(-1)^{|e|} [\tilde{z}_{e_{1}} | z_{e_{2}}\ket [\tilde{z}_{e_2}| z_{e_3}\ket \cdots [\tilde{z}_{e_{n}}|z_{e_1}\ket
\ee
where $\tilde{z}_e \equiv z_{e^{-1}}$ and $|e|$ is the number of edges in the cycle which agrees with the orientation of $\Gamma$.  
Expanding the generating functional we obtain the explicit expansion
\be
{\cal A}_{\Gamma}(z_{e}) = \sum_{k_{ee'}^{v}} R_{\Gamma}(k_{ee'}^{v}) \prod_{v} \left(\prod_{(ee')\supset v} [z_{e}|z_{e'}\ket^{k_{ee'}^{v}}\right).
\ee
where $R_{\Gamma}(k_{ee'}^{v})$ are the generalization of the Racah summation for an arbitrary graph
\be \label{eqn_power}
R_{\Gamma}(k_{ee'}^{v}) \equiv  \sum_{[M_C]} (-1)^{N+s} \frac{(N+1)!}{\prod_{C} M_C!}
%  \frac{1}{(1+\sum_{C} A_C(z_e))^2} = \sum_{[M_C]} (-1)^{N}  \prod_{C} \frac{(N+1)!}{M_C!} A_{C}(z_e)^{M_C}= \sum_{[k_{ee'}]} A_\Gamma(k_{ee'}) \prod_{v \in V_\Gamma} 
%\frac{ \bra z_{e} | [k_{ee'}^{v}]\ket }{||[k_{ee'}^{v}]||^{2}}
\ee
where $N = \sum_C M_C$ and the sign $s$ accounts for the ordering of $ee'$ in $[z_{e}|z_{e'}\ket$.  The $M_C$ are positive integers labelled by each disjoint union of cycles $C$ and are summed over. These integers are restricted to depend on the $k_{ee'}^{v}$ by the relation
\be \label{eqn_kee}
  k_{ee'}^{v} = \sum_{C \supset (ee')} M_C,
\ee
where the sum is over all cycle unions $C$ which contain a cycle with the corners $(ee')$ or $(e'e)$.\footnote{Note that the solution of (\ref{eqn_kee}) is not unique since in  the number of cycles is usually greater than the number of independent $k_{ee'}$.  Therefore in general the coefficients $A_\Gamma(k_{ee'}^{v})$ will be given by a sum over arbitrary parameters. This leads  for example  to a summation over  one parameter for the tetrahedral graph, which corresponds to  the Racah expansion of the 6j symbol.  For the 4-simplex this will involve 17 parameters.}

On the other hand the relationship between continuous and discrete bases implies that the generating functional can also be expressed in terms of the discrete intertwiners as
\be
{\cal A}_{\Gamma}(z_{e}) = \sum_{k_{ee'}^{v}} A_{\Gamma}(k_{ee'}^{v}) \prod_{v} \left(\prod_{(ee')\supset v} [z_{e}|z_{e'}\ket^{k_{ee'}^{v}}\right).
\ee
This shows that  
$A_{\Gamma}(k_{ee'}^{v}) \simeq R_{\Gamma}(k_{ee'}^{v})$ 
where $\simeq$ is an equivalence relation on amplitudes $A_{\Gamma}(k^{v}_{ee'})$.
It is defined by  $A_{\Gamma}(k^{v}_{ee'})\simeq 0$ iff $\sum_{k^{v}_{ee'}}A_{\Gamma}(k^{v}_{ee'}) \prod_{v,(ee')} [z_{e}|z_{e'}\ket^{k_{ee'}^{v}}=0$. That is, it vanishes due to the Plucker relations when contracted with and summed over $\prod_{v,(ee')} [z_{e}|z_{e'}\ket^{k_{ee'}^{v}}$.

In order to find the  analog of the Racah formula  for the amplitude $A_{\Gamma}(k_{ee'}^{v})$ we need to use the more general generating functional 
\be
{\cal G}_{\Gamma}(\tau^{v}) \equiv \sum_{k_{ee'}^{v}} A_{\Gamma}(k_{ee'}^{v}) \prod_{v} \left(\prod_{(ee')\supset v} (\tau_{ee'}^{v})^{k_{ee'}^{v}}\right)
\ee
where $\tau^{v}_{ee'}$ are arbitrary complex parameters associated with pairs of edges meeting at $v$.
This generating functional has been evaluated in \cite{FH_exact}, and remarkably it was shown that they assume the same form as (\ref{eqn_gen_func}), the only difference being that the sum is {\it not only} over unions of cycles, but also over unions of simple loops denoted by $L$.
A simple loop is loop of non overlapping edges.
The difference between loops and cycles is that cycles do not have any intersections.  Hence, the unions of cycles are a subset of the unions of loops.  This result implies that the amplitude $A_{\Gamma}(k_{ee'}^{v})$ can be expressed as a Racah sum over loops:
\be %\label{eqn_power}
A_{\Gamma}(k_{ee'}^{v}) \equiv  \sum_{[M_L]} (-1)^{N+s} \frac{(N+1)!}{\prod_{L} M_L!}
\ee
where $M_{L}$ are integers labelled by each disjoint union of simple loops $L$, 
they are summed over with the restriction
$
  k_{ee'}^{v} = \sum_{L \supset (ee')} M_L,
$
while
$N = \sum_L M_L$.  For more details see \cite{FH_exact}.

\subsection{The 20j symbol}

Let us now use all of the results obtained for the $|S,T\ket$ basis to compute a generalization of the $15j$ symbol, which will depend now on 20 spins: ten $j_e$ on the edges and five $S_v$, and five $T_v$ on the vertices.  The $20j$ symbol will be the amplitude corresponding to the coherent 4-simplex.  It is a generalization of the $15j$ symbol since by theorem \ref{thm_sum_T} we can sum over five of the extra spins and obtain one of the five variations of the $15j$ symbol.

First label the vertices of the 4-simplex by $i=1,..,5$ and an edge directed from $i$ to $j$ as in $z_{e} \equiv z^{i}_{j} \neq z^{j}_{i} \equiv z_{e^{-1}}$. Let $\{20j\}_{S_i,T_i} \equiv \underset{i}{\corner} |S_i, T_i \ket$ denote the unnormalized 20j symbol.
The relation (\ref{eqn_discrete_amp}) between coherent and discrete amplitudes reads
\be 
  A_{4S}(j_{ij},z^{i}_{j}) = \sum_{S_i,T_i} \{20j\}_{S_i,T_i} \prod_{i} \frac{( z^{i}_{j}  |  S_i, T_i\ket}{\|S_i,T_i\|^2}
\ee
\begin{figure} 
  \centering
    \includegraphics[width=0.5\textwidth]{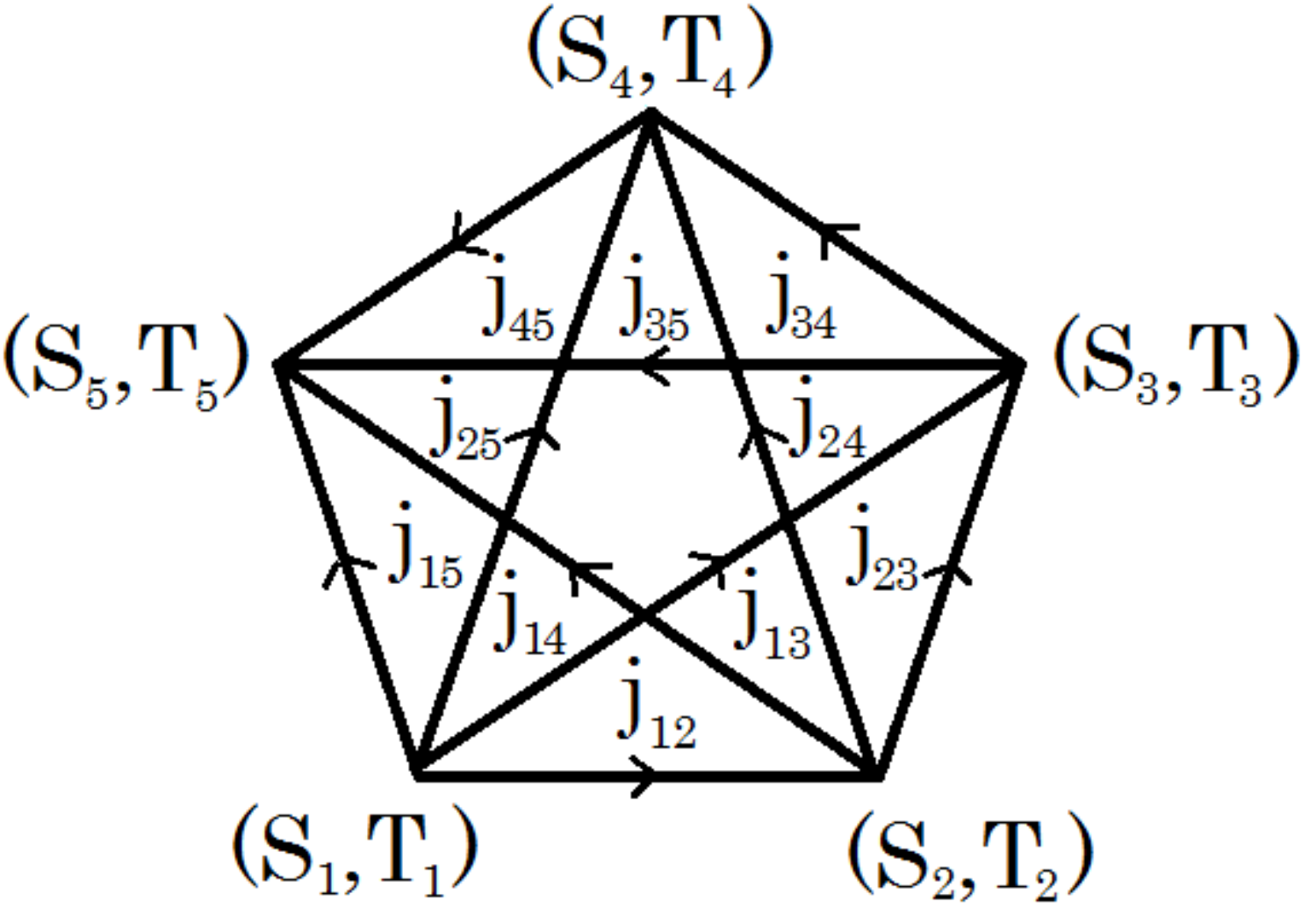}
    \caption{The graphical representation of the 20j symbol: The amplitude of the coherent 4-simplex.}  \label{fig_20j}
\end{figure}
We can express the $20j$ symbol explicitly in terms of the $15j$ symbol by inserting another five resolutions of identity $\one_{j_i} = \sum_{S'} |S'\ket\bra S'|/\|S'\|^2$ into the definition of the 20j symbol to get
\be \label{eqn_20j_15j}
  \{20j\}_{S_i,T_i} = \sum_{S_{i}^{'}} \{15j\}_{S'_i} \prod_{i} \frac{\bra S_{i}^{'} | S_i, T_i \ket}{\|S_{i}^{'}\|^2}
\ee
where $\{15j\}_{S'_i}$ is the unnormalized 15j symbol defined by $\{15j\}_{S'_i} \equiv \underset{i}{\corner} |S_{i}^{'} \ket$ and is equal to $\prod_i \|S_{i}^{'}\|$ times the conventional normalized 15j symbol up to a sign depending on the orientation of the edges.  Notice that by summing over $T_i$ in (\ref{eqn_20j_15j}) we obtain the unnormalized $15j$ symbol as expected.  Thus the five different kinds of $15j$ symbols are derived from the $20j$ by summing over the different channels.  For example the $15j$ with all $S$ channels is given by
\be
  \{15j\}_{S_i} = \sum_{T_i} \{20j\}_{S_i,T_i}, 
\ee
and the other kinds of 15j symbol are given similarly.

Let us now use theorem \ref{thm_sum_T} to rewrite the 20j symbol in a more symmetric form
\be
  \{20j\}_{S_i,T_i} = \sum_{S_{i}^{'},T_{i}^{'}} \{15j\}_{S'_i} \prod_{i} \frac{\bra S_{i}^{'},T_{i}^{'} | S_i, T_i \ket}{\|S_{i}^{'}\|^2}.
\ee
In this form it is easy to derive the asymptotics of the 20j symbol by those of the 15j since for large spins (see the next section) $\bra S,T | S', T' \ket \sim \delta_{S,S'} \delta_{T,T'} \|S,T\|^2$, and therefore 
\be
  \{20j\}_{S_i,T_i}\sim  \{15j\}_{S_i} \prod_i \frac{\|S_i,T_i\|^2}{\|S_i\|^2}.
\ee
This means that understanding the asymptotics of the $20j$ symbol will give us the asymptotics of the $15j$ symbol too.
There has been recent results on the asymptotics of spin networks evaluation \cite{FC,B} but this progress concerns however the asymptotic evaluation of the coherent state amplitude $ A_{4S}(j_{ij},z^{i}_{j})$. 
The asymptotic evaluation of the non coherent $15j$ symbol is not known and as we are going to see in the next section our techniques allow us to unravel the asymptotics for the first time.

Finally, we will give an explicit expression for the 20j, which is independent of the 15j, as a generalized Racah formula.  By solving (\ref{eqn_kee}) for $M_C$ in terms of $k^{v}_{ee'}$ we can derive a Racah formula for the amplitude of an arbitarary graph which is given by (\ref{eqn_power}).  Since there are 37 cycles $C$ in the 4-simplex and only 20 independent $k^{v}_{ee'}$ this formula will not be unique and will involve a sum over 17 parameters $p_k$.  The Racah formula is then
\be
  \{20j\}_{S_i,T_i} \simeq \sum_{p_1 \cdots p_{17}} \frac{(-1)^{N+s}(N+1)!}{\prod_{C} M_{C}(j_{ij},S_i,T_i,p_k)!}
\ee
where $N = \sum_C M_C$ and the sign $s = M_{1234}+M_{1235}+M_{1245}+M_{12354}+M_{12435}$ accounts for the edge ordering.  In appendix \ref{20j_symbol} we give an explicit parameterization of the $M_C$ in terms of the $p_k$ although we note that simpler parametrizations might exist.  Furthermore, using various hypergeometric formulas one may be able to perform some of the summations over the $p_k$ explicitly.

\section{Semi-Classical Limit}

It is now well-known and explained in great detail in \cite{Conrady:2009px,Freidel:2009nu} that the space of $4$-valent intertwiners can be uniquely labelled by 
oriented tetrahedra. In this section we will demonstrate this correspondence for the $|S,T\ket$ states.  In order to connect with the classical behaviour we would like to analyze the asymptotics of the scalar product of two such states in the limit where the spins $(j_{i},S,T)$ are all uniformly large.
We use the fact that this scalar product can itself be expressed as an integral 
\be
\la S,T|S',T'\ra =\frac1{\prod_{i<j}( k_{ij}! k_{ij}'!)} \int \prod_{i} \frac{\rd^2 z_{i}}{\pi^2}  e^{- S_{k}(z)}
\ee
  where the action is given by
  \be \label{eqn_action}
  S_{k} = \sum_{i} \la z_{i}| z_{i} \ra - \sum_{i<j}\left( k_{ij} \ln[z_{i}|z_{j}\ra + k'_{ij} \ln \la z_{i}|z_{j}] \right).
  \ee
The asymptotic evaluation of this scalar product is controlled by the stationary points\footnote{ If $k_{ij}= N K_{ij}$ and we  define $ z_{i}=\sqrt{N} x_{i}$ we see that this integral that we want to evaluate in the large $N$ limit takes the usual form $N^{2J} \int \prod_{i}\rd x_{i} e^{-N S_{K}(x)} $ .} of this action.
That is we look for solutions of 
\be
\sum_{j\neq i} \frac{k_{ij}}{[z_{i}|z_{j}\ra} [z_{i}| = \la z_{j} |,\qquad
\sum_{j\neq i} \frac{k_{ij}'}{\bra z_{j}|z_{i}]} |z_{i}] = | z_{j} \ra .\label{kz}
\ee
Now it is clear that if $k\neq k'$ there cannot be any real solution. This shows that this scalar product is exponentially suppressed unless $(S,T)=(S',T')$
\footnote{We could still evaluate the  integral asymptotically  when  $k\neq k'$ by looking for complex solutions.
In order to do so we  and use the fact that $[\check{z}_{i}|\check{z}_{j}\ra = [z_{i}|z_{j}\ra^{*}$.
We get an action holomorphic in $ |z_{i}\ket$ and $|\check{z}_{i}\ket$:
$$S_{k} =- \sum_{i} [ \check{z}_{i}| z_{i} \ra - \sum_{i<j}\left( k_{ij} \ln[z_{i}|z_{j}\ra + k'_{ij} \ln [\check{ z}_{i}|\check{z}_{j}\ra \right).$$
 The stationary equations are 
\be
\sum_{j\neq i} \frac{k_{ij}}{[z_{i}|z_{j}\ra} [z_{j}| = -[ \check{z}_{i} |,\qquad 
\sum_{j\neq i} \frac{k_{ij}'}{[\check{z}_{i}|\check{z}_{j}\ra} [\check{z}_{j}| =  [ z_{i} |,
\ee
In the case $k_{ij}\neq k'_{ij}$ we do not demand that $[{z}_{i}| = \la \check{z}_{i}|$ which corresponds to the real contour of integration.}.
Furthermore, if we contract this   equation with $ | z_{j}\ra$  we obtain the constraints
\be
2j_{i} = \sum_{j\neq i} k_{ij}= \la z_{i}|{z}_{i}\ra.
\ee
These equations are invariant under $\SU(2)$,  so $  g|z_{i}\ra$ is a solution if $|z_{i}\ra$ is and $ g \in \SU(2)$.
We also have an invariance of these equations under the rescaling $ |z_{i}\ra \to e^{i\alpha^{i}} |z_{i}\ra$.

Finally, by taking the conjugation of (\ref{kz}) $|\check z\ket = |z]$ and using the fact that $[\check{z}_{j}|\check{z}_{i}\ra = [z_{j}|z_{i}\ra^{*}$ we can show that this equation is also equivalent
to the conjugated equation 
\be
\sum_{j\neq i} \frac{k_{ij}}{[\check{z}_{j}|\check{z}_{i}\ra} [\check{z}_{j}| = \la \check{z}_{i} |.\label{conjkz}
\ee
This means that the $\mathbb{Z}_{2}$ transformation  $ |z_{i}\ra \to | \check{z}_{i}\ra =|z_{i} ]$ is also a symmetry of the equation of motion.  In summary this shows that the symmetry group of the solutions (\ref{kz}) is given by $\SU(2) \times \mathrm{U}(1)^{4}\times \mathbb{Z}_{2}$.

\subsection{Relation with Framed Tetrahedra}

What is remarkable about the solutions (\ref{kz}) is that they are in one to one correspondence with framed tetrahedra.
A framed tetrahedron in $\R^{3}$ is a tetrahedron together with a choice of frame on each face (i.e. a choice of a preferred direction tangential to the face).
The SU(2) invariance corresponds to rotations of the tetrahedron, while a rotation of the frame on face $i$ by an angle $\alpha^{i}$ corresponds to a rescaling of $|z_{i}\ra $ by $ e^{i\alpha^{i}/2}$.
The $\mathbb{Z}_{2}$ transformation corresponds to a global reflection exchanging inward and outward normals.

Indeed, suppose that we have a framed tetrahedron which is  such that the area and outward unit normal directions of the face $i$ are  denoted by $(A_{i},N_{i})$.
We also denote $F_{i}$ to be the unit vector in the face $i$ (i.e. $F_{i}\cdot N_{i}=0$) that provides the framing of the face $i$.
Then the fact that this data corresponds to a tetrahedron is implied by the closure constraints
\be
\sum_{i} A_{i} N_{i} =0.
\ee

Such a framed tetrahedron can be equivalently labelled in terms of four spinors $|z_{i}\ra$ which satisfy the closure relation 
\be
\sum_{i} |z_{i}\ket \bra z_{i}| = \frac{A}2 1
\ee
where $A=\sum_{i}A$ is  the total area of the  tetrahedra.
This data is 
related to the data $(A_{i},N_{i},F_{i})$ as follows: First $\bra z_{i}|z_{i}\ket = A_{i}$ and second
\be \label{eqn_z_N_F}
|z_{i}\ra\la z_{i}| -|z_{i}][z_{i}| = A_{i} N_{i} \cdot \sigma ,\qquad  |z_{i}\ra[z_{i}| = i\frac{A_{i}}{2}\left(F_{i} + i N_{i}\times F_{i}\right) \cdot \sigma
\ee
where $\sigma = (\sigma_{1},\sigma_{2},\sigma_{3})$ are the Pauli matrices and $\times$ denotes the cross product.
The first equation determines $|z_{i}\ra$ up to a phase while the second equation determines the phase up to an overall sign. Thus $\pm |z_{i}\ra$ is uniquely determined by the framed tetrahedron.

\begin{theorem} \label{thm_closure}
The solutions of (\ref{kz}) are in one to one correspondence with framed tetrahedra, with face areas $2j_{i}$ and total area $2J$.
 The discrete parameters related to the spinors are given by
\be
\bra z_{i}|z_{i}\ket = 2j_{i}, \quad J k_{ij} \equiv   |[z_{j}|z_{i}\ket|^{2}.
\ee
%
%The correspondence between solutions of (\ref{kz}) and framed tetrahedra is unique up to the $\SU(2) \times \mathrm{U}(1)^{4}\times \mathbb{Z}_{2}$ gauge transformations.  To be precise the solutions of (\ref{kz}) are equivalent to the closure condition for a tetrahedron
%\be \label{eqn_closure}
%\sum_{j} |z_{j}][ z_{j}| = J 1
%\ee
%where $J =\sum_{i} j_{i}$ is half the total area of the tetrahedron.
\end{theorem}
\begin{proof}
Lets suppose that $|z_{i}\ket$ is a solution of (\ref{kz}).  Then
\bea
\sum_{i} |z_{i}\ra \la z_{i} | = 
\sum_i \sum_{j\neq i } \frac{k_{ij}}{[z_{j}|z_{i}\ra} |z_{i}\ra [z_{j}| 
= \sum_{i < j } \frac{k_{ij}}{[z_{j}|z_{i}\ra} (|z_{i}\ra [z_{j}|  -|z_{j}\ra [z_{i}|)
= \sum_{i < j } {k_{ij}} 1 = J 1.
\eea
which implies that $|z_{i}\ket$ satisfy the closure constraint and hence constitute a framed tetrahedron.
The area of the faces of this tetrahedron are given by $A_{i} =\sum_{j\neq i} k_{ij}$.

Let us now suppose that $|z_{i}\ket$ is a solution of the closure constraints and lets define $k_{ij} \equiv \frac2{A} [z_{j}|z_{i}\ket\bra z_{i}|z_{j}]$.
Then by construction we have 
\be
\sum_{j\neq i} \frac{k_{ij} }{[z_{j}|z_{i}\ket} [z_{j}| = \frac2{A} \sum_{j\neq i}\bra z_{i}|z_{j}] [z_{j}| = \bra z_{i}|.
\ee
which shows that $|z_{i}\ket$ is a solution of (\ref{kz}). 

Finally lets suppose that given $|z_{i}\ket$ we have another set $k'_{ij}$ which is a solution of (\ref{kz}).
 This would imply that 
$\sum_{j\neq i} {\Delta_{ij}} [z_{j}| =0$ with $\Delta_{ij}\equiv(k_{ij}-k_{ij}')/[z_{j}|z_{i}\ket$.
The sum contains three terms, and by contracting it with $|z_{k}]$, with $k\neq i,j$, we obtain  two relations.
The consistency of these relations imply that $\Delta_{ij}=0$. 
\end{proof}

This shows that the $|S,T\ket$ states enjoy the same geometrical properties as the coherent intertwiners.  Namely they are peaked on states representing closed bounded tetrahedra.  We note that the proof of theorem \ref{thm_closure} also holds for general $n$-valent intertwiners by simply extending the range of indices from 4 to $n$.  A similar analysis of stationary points of intertwiner generating functionals is given in \cite{Bonzom:2012bn}.

\subsection{Geometrical Interpretation}
It is interesting to make explicit the geometrical interpretation of the data encoded in the spinor variables.
We have seen that the norms of the spinors $2j_{i}=\bra z_{i}|z_{i}\ket$ are the areas of the faces. %We are going to see that $k_{ij}$ is proportional to the length of the edge $(ij)$.

This can be made explicit by writing these spinors in terms of the geometrical data:
As we have seen $A_{i}$ denotes the area of the face $i$ and 
we denote by $\theta_{ij} \in [0,\pi]$ the dihedral angle between the normals $N_{i}$ and $N_{j}$.
The extra data necessary is the angle $\alpha^{i}_{j}$ in the face $i$ between   the oriented edge $(ij)$ and the reference vector $F_{i}$.
This data is represented in figure \ref{fig_tri} and is related to the spinors in the following lemma.

\begin{figure} 
  \centering
    \includegraphics[width=0.7\textwidth]{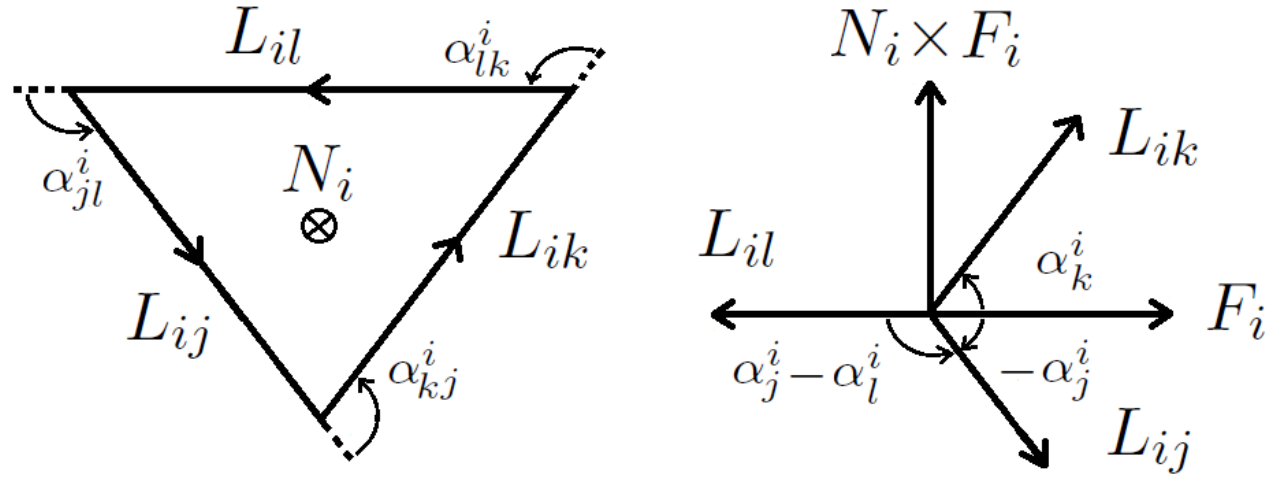}
    \caption{The geometrical data on the face $i$ of a framed tetrahedron.}  \label{fig_tri}
\end{figure}

\begin{lemma} \label{lemma_geo}
The angle between the edge $(ij)$ and the edge $(ik)$ is $\alpha^{i}_{jk}$ where
$\alpha^{i}_{jk} = \alpha^{i}_{j}-\alpha^{i}_{k}$.
The expression of the spinor products in terms of the geometrical data is  given by:
\bea
[z_{i}|z_{j}\ra &=& \epsilon_{ij} \sqrt{A_{i}A_{j}} \sin \frac{\theta_{ij}}{2} e^{i (\alpha^{i}_{j} + \alpha^{j}_{i})/2},\\
{[}z_{i}|z_{j}] &=& \sqrt{A_{i}A_{j}} \cos \frac{\theta_{ij}}{2} e^{i (\alpha^{i}_{j} - \alpha^{j}_{i})/2},
\eea
where $\epsilon_{ij} = +1$ if $(ij)$ is positively oriented and $-1$ otherwise.  The area of face $i$ is $A_{i} = 2j_{i}$ and
$\theta_{ij}\in[0,\pi]$ is the 3d external dihedral angle for which we choose the convention $\theta_{ii}=0$.  Finally, $\alpha^{i}_{j}$ is the angle in the face $i$ between the edge $(ij)$ and the reference vector $F_{i}$.
\end{lemma}
\begin{proof}
%Let us first note that the change of frame $|z_{i}\ket \to e^{i\beta^{i}}|z_{i}\ket $ amounts to a rescaling $ \alpha^{i}_{j} \to \alpha^{i}_{j} + \beta^{i}$.
From the definitions (\ref{eqn_z_N_F}) we have
\be
|z_{i}\ket\bra z_{i}|=  \frac{A_{i}}{2}(1 + N_{i}\cdot \sigma ),\qquad 
|z_{i}][ z_{i}|= \frac{A_{i}}{2}(1 - N_{i}\cdot \sigma ).
\ee
and so the scalar product between two normals is
\be \label{eqn_N_dot}
  A_i A_j N_i \cdot N_j = |\bra z_i | z_j \ket|^2 - |\bra z_i | z_j]|^2 = A_i A_j \cos \theta_{ij}.
\ee
Defining the edge vectors by $L_{ij}\equiv A_{i}A_{j}(N_{i}\times N_{j})$ then gives
\be\label{Lij}
 L_{ij} \cdot \sigma = 2i \Big( |z_{i}\ra\la z_{i} | z_{j}][z_{j}| - |z_{j}][z_{j}|z_{i}\ra \la z_{i}| \Big).
\ee
%=-2i A_{i}A_{j} (N_{i}\times N_{j})\cdot \sigma.
Now by taking the trace of the square of (\ref{Lij}) we obtain 
\be \label{eqn_N_cross}
  A_i A_j |N_i \times N_j | =  2|[z_{i}|z_{j}\ket \bra z_j | z_i \ket |=  A_i A_j \sin \theta_{ij}.
\ee
Equations (\ref{eqn_N_dot}) and (\ref{eqn_N_cross}) determine the magnitudes of the spinor products.  %Furthermore, the identification $k_{ij} \equiv \frac2{A} [z_{j}|z_{i}\ket\bra z_{i}|z_{j}]$ from theorem \ref{thm_closure} implies that on shell
%\be
%  k_{ij} = \frac{|L_{ij}|}{A} 
%\ee
%and by construction $|L_{ij}|$ is the volume of the tetrahedron times the length of edge $(ij)$.

Lets now look at the scalar product between the edge vectors $L_{ij}$ %\equiv A_{i}A_{j}(N_{i}\times N_{j})$ 
and the complex vector $F_{i} + i N_{i}\times F_{i}$
%Using again (\ref{Lij})  the expression of $F$ in terms of the spinors this is given by
\bea
 L_{ij}\cdot (F_{i} + i N_{i}\times F_{i})& =&
\frac{2}{A_{i}}\tr\left( \left\{|z_{i}\ra\la z_{i} | z_{j}][z_{j}| - |z_{j}][z_{j}|z_{i}\ra \la z_{i}| \right\} 
 |z_{i}\ra [ z_{i}| \right)  \label{eqn_L_F_N} \\
&= & \frac{2 }{A_{i}}  [ z_{i}|z_{j}][  z_{i} | z_{j}\ra \la z_{i}|z_{i}\ra   \nn \\
&=& \epsilon_{ij} A_{i}A_{j}\sin {\theta_{ij}} e^{i\alpha^{i}_{j} } = \epsilon_{ij} |L_{ij}| e^{i\alpha^{i}_{j}} \nn
\eea
where $|L_{ij}| = A_i A_j \sin \theta_{ij}$.  This shows that $L_{ij} \cdot F_i = |L_{ij}| \cos(\alpha^{i}_{j})$ and so $\alpha^{i}_{j}$ is indeed the angle  between the edge $(ij)$ and the frame vector on face $i$.  The sign of $L_{ij}\cdot (N_{i}\times F_{i})$ determines the orientation of this angle with respect to the 2d basis $\{F_i,\hat{F}_{i}\}$ where $\hat{F}_i \equiv N_{i}\times F_{i}$.

We can also show that the angle $\alpha^{i}_{jk} \equiv \alpha^{i}_{j} - \alpha^{i}_{k}$ is the angle between the edge vectors $L_{ij}$ and $L_{ik}$ at the vertex $i$.
Using (\ref{eqn_L_F_N}) we can construct the following quantity
\be
  \epsilon_{ij} \epsilon_{ki} |L_{ij}| \cdot |L_{ik}| e^{i(\alpha^{i}_{j} - \alpha^{i}_{k})} = [L_{ij} \cdot (F_{i} + i N_{i} \times F_{i})] \cdot [L_{ik} \cdot (F_{i} - i N_{i} \times F_{i})].
\ee
Now in the components of the 2d basis $L^{(1)}_{ij} = L_{ij} \cdot F_{i}$ and $L^{(2)}_{ij} = L_{ij} \cdot \hat{F}_{i}$ we have
\be
  \epsilon_{ij} \epsilon_{ki} |L_{ij}| \cdot |L_{ik}| e^{i(\alpha^{i}_{j} - \alpha^{i}_{k})} 
  = (L^{(1)}_{ij} + i L^{(2)}_{ij})(L^{(1)}_{ik} - i L^{(2)}_{ik}).
  %= L_{ij} \cdot L_{ik} - i(L_{ij} \times L_{ik}).
\ee
The real part is equal to $L_{ij} \cdot L_{ik}$.  % and the imaginary part is the cross product of the edges and gives the orientation of this angle in the 2d basis.  
Therefore
\be \label{eqn_alpha_i_jk}
  L_{ij} \cdot L_{ik} = \epsilon_{ij} \epsilon_{ki} |L_{ij}| \cdot |L_{ik}| \cos \alpha^{i}_{jk}.
\ee
which shows that $\alpha^{i}_{jk}$ is the angle between edges $(ij)$ and $(ik)$.
\end{proof}

Using (\ref{eqn_alpha_i_jk}) and the definition $L_{ij} = A_i A_j (N_i \times N_j)$ we can relate the angles $\alpha^{i}_{jk}$ to the 3d dihedral angles by
\be \label{eqn_alpha_theta}
  \epsilon_{ij} \epsilon_{ki} \cos \alpha^{i}_{jk} = \frac{\cos \theta_{jk} - \cos \theta_{ij} \cos \theta_{ik}}{ \sin \theta_{ij} \sin \theta_{ik}}.
\ee
This is the spherical law of cosines relating the edges $(\theta_{ij},\theta_{jk},\theta_{ki})$ and angles 
$(\alpha_{ji}^{k},\alpha_{kj}^{i},\alpha_{ik}^{j})$ of a spherical triangle.  This relation with spherical geometry is captured by the so called three terms relations which we discuss next.

\subsection{Geometry of 3-terms Relations}

The relationships between the the 3d dihedral angles and the internal angles between edges is expressed via the three term relations 
(Fierz identity) satisfied by a set of spinors.
There are two such types of relations.
First there are the relations arising at a given vertex of the tetrahedra which imply that the angles $(\theta_{ij},\theta_{jk},\theta_{ki})$ and 
$(\alpha_{ji}^{k},\alpha_{kj}^{i},\alpha_{ik}^{j})$ are respectively the edge lengths and angles of a spherical triangle.
We can write two such relations\footnote{Note that $|z_{j}\ra \la z_{j} | + |z_{j}][ z_{j} | = \la z_{j} |z_{j}\ra \one$.}, 
\bea
\la z_{i}|z_{j}\ra \la z_{j} |z_{k}\ra + \la z_{i}|z_{j}][ z_{j} |z_{k}\ra &=& \la z_{i}|z_{k}\ra \la z_{j} |z_{j}\ra \\
\la z_{i}|z_{j}\ra \la z_{j} |z_{k}] + \la z_{i}|z_{j}][ z_{j} |z_{k} ] &=& \la z_{i}|z_{k}] \la z_{j} |z_{j}\ra 
\eea
which translate into
\bea
%\cos  \frac{\theta_{ij}}2 \cos  \frac{\theta_{jk}}2 e^{i(\alpha^{j}_{ik} + \alpha^{i}_{kj} +\alpha^{k}_{ji})/2}
%+ \sin  \frac{\theta_{ij}}2 \sin  \frac{\theta_{jk}}2 e^{i( \alpha^{i}_{kj} - \alpha^{j}_{ik}  + \alpha^{k}_{ji})/2}
%=\cos \frac{\theta_{ik}}{2}
%\cos  \frac{\theta_{ij}}2 \cos  \frac{\theta_{jk}}2 
%+ e^{i\alpha^{j}_{ik} }  \sin  \frac{\theta_{ij}}2 \sin  \frac{\theta_{jk}}2 
%=\cos \frac{\theta_{ik}}{2} e^{i(\alpha^{j}_{ik} + \alpha^{i}_{kj} +\alpha^{k}_{ji})/2},
c_{ij} c_{jk} \label{eqn_3_term_1}
+ s_{ij} s_{jk} e^{i\alpha^{j}_{ki} }
&=& c_{ik} e^{i(\alpha^{i}_{jk}+\alpha^{j}_{ik} + \alpha^{k}_{ij})/2},\\
c_{ij} s_{jk} 
+  s_{ij} c_{jk} e^{i\alpha^{j}_{ki} }
&=&  s_{ik} e^{i( \alpha^{i}_{jk}-\alpha^{j}_{ik} - \alpha^{k}_{ij})/2}, \label{eqn_3_term_2}
\eea
where $c_{ij} \equiv \cos  \frac{\theta_{ij}}2$ and $ s_{ij}\equiv \epsilon_{ij}\sin  \frac{\theta_{ij}}2$.  Taking the difference of squares of equations $|(\ref{eqn_3_term_1})|^2 - |(\ref{eqn_3_term_2})|^2$ produces equation (\ref{eqn_alpha_theta}).

We also have a relation that  genuinely depends on the tetrahedral geometry and involves the four spinors;
it follows from the Pl\"ucker relation that we have already made extensive use of
\be
[z_{1}|z_{2}\ket[z_{3}|z_{4}\ket + [z_{1}|z_{3}\ket[z_{4}|z_{2}\ket + [z_{1}|z_{4}\ket[z_{2}|z_{3}\ket =0
\ee
and it reads 
\be
s_{12}s_{34} e^{i (\alpha_{12} +\alpha_{34})/2} + s_{13}s_{24} e^{i (\alpha_{13}+\alpha_{24})/2} + s_{14}s_{23}e^{i(\alpha_{14}+\alpha_{23})/2} =0
\ee
where we have defined 
%\be
%\alpha^{(ij)(kl)} \equiv \sum_{b} \alpha^{i}_{jb} + \alpha^{j}_{ib} + \alpha^{k}_{lb}+ \alpha^{l}_{lb}
%\ee
\be \label{eqn_alpha_geo}
\alpha_{ij}\equiv \frac12 \sum_{k\neq i,j} (\alpha^{i}_{jk} +\alpha^{j}_{ik}).
\ee
It can be checked that these angles sum up to $0$: $\sum_{i\neq j } \alpha_{ij}=0$.  Now what needs to be appreciated is the non trivial fact that the angles 
$$
\Phi_{S} = \alpha_{12} +\alpha_{34}
,\quad
\Phi_{T} =\alpha_{13}+\alpha_{24},\quad
\Phi_{U} = \alpha_{14}+\alpha_{23}
$$
determine completely the geometry of the tetrahedron once we know the face areas $A_{i}=2j_{i}$.
This means that $(j_{i}, \alpha_{S},\alpha_{T})$ determine the value of all the 3d dihedral angles $\theta_{ij}$ and internal angles $\alpha_{ij}^{k}$.
This non-trivial fact follows from the analysis performed in \cite{EKLF}.

\subsection{Asymptotic Evaluation of the 20j Symbol}

We will now take an indepth look at the asymptotic evaluation of the normalized  $20j$ symbol.
This object depends on the choice of an orientation of the edges, and we denote by $\epsilon_{ij}$ a sign which $+1$ if the edge $[ij]$ is positively oriented from $i$ to $j$ and $-1$ otherwise. 
This normalized $20j$ symbol is defined as a contraction of the normalized intertwinner $ |S,T\ra$ times the normalisations  and it is 
 expressed as 
\be \label{normalized_20j}
\widehat{\{20j\}}_{S_a,T_a}\equiv \frac{\{20j\}_{S_a,T_a}}{\prod_a \|S_a,T_a\|} = \frac{I(k_{ij})}{\sqrt{\prod_{a}(J_{a}+1)! \prod_{a\neq i<j}k^{a}_{ij}! }} ,
\ee
where $I(k_{ij})$ is an integral over $20$ spinors $|z_{i}^{j}\ra $. The countour of integration is a real contour where $ |z_{i}^{j}\ra$ is related to the  conjugate $|z_{j}^{i}]$ by the reality condition.
\be \label{eqn_reality_cond}
|z_{i}^{j}\ra = \epsilon_{ij}|z_{j}^{i} ].
\ee
This condition implies that the normals of glued faces are related by $N^{j}_{i} = - N^{i}_{j}$ and that the frame vectors match $F^{i}_{j} = F^{j}_{i}$.  
%\be
%  A_{ij} N^{j}_{i} \cdot \sigma = |z^{j}_{i}\ket \bra z^{j}_{i}| - |z^{j}_{i}] [ z^{j}_{i}| =  |z^{i}_{j}] [ z^{i}_{j}| - |z^{i}_{j}\ket \bra z^{i}_{j}| = -A_{ij} N^{i}_{j} %\cdot \sigma
%\ee

The integral  is given by
\be
I(k_{ij}) =\int \prod_{i\neq j} \frac{\rd^{2} z_{j}^{i} }{\pi^2}  e^{S(z^{i}_{j}) },\quad \mathrm{with}\quad S \equiv \sum_{i<j}  [ z^{i}_{j} | z_{i}^{j}\ket + \sum_{a} \sum_{i<j} k_{ij}^{a} \ln [z_{i}^{a}|z_{j}^{a}\ra
\ee
There are four spinors $|z_{i}\ra$ associated with a framed tetrahedron.
The stationary points of this equation are given by solutions of 
\be \label{system_of_equations}
\sum_{j\neq a,i} \frac{ k_{ij}^{a} }{[z_{j}^{a}|z_{i}^{a}\ra} [z_{j}^{a}| = - [z^{i}_{a}| = \epsilon_{ai} \la z_{i}^{a}|,
\ee
and according to the previous section the solution of these equations are given by
oriented framed tetrahedra.
The relationship between $k_{ij}^{a}$ and the spinors depends on the choice of graph orientation,
it is given by
\be \label{eqn_k_z}
k_{ij}^{a} = \frac1{J^{a}} |[\hat{z}_{i}^{a}|\hat{z}_{j}^{a}\ra|^{2}, \quad 2j_{ai}=  \la \hat{z}^{a}_{i}|\hat{z}^{a}_{i}\ra,\quad J^{a} = \sum_{i\neq a} j^{a}_{i} 
\ee
where $|\hat{z}^{a}_{i}\ra =|z^{a}_{i}\ra$ if the edge is oriented from $a$ to $i$ and 
$|\hat{z}^{a}_{i}\ra = |z^{a}_{i}]$ if the edge is oriented from $i$ to $a$.
This determines the norm of the spinor scalar products in terms of $k_{ij}^{a}$. 

The phases of these products are denoted $\alpha^{ab}_{i}$ and they denote the angle in the face $b$ of the tetrahedron $a$, between the edge $(bi)$ and the reference frame vector in the face $b$ of tetrahedra $a$.
As shown in lemma \ref{lemma_geo}, they are related to the spinor products by 
\be
[\hat{z}_{i}^{a}|\hat{z}_{j}^{a}\ra =\sqrt{J^{a} k_{ij}^{a}} e^{i (\alpha_{j}^{ai}+ \alpha_{i}^{aj})/2}.
\ee

Thus the on-shell evaluation of the action is
\be
  S_{\mathrm{onshell}} = -\sum_{a} J^{a} +\frac12 \sum_{a\neq i<j} k_{ij}^{a}\ln (J^{a} k_{ij}^{a} ) + \frac{i}{2}\sum_{a\neq i<j} k_{ij}^{a}  (\alpha_{i}^{aj} + \alpha_{j}^{ai})
\ee
The real part can be rewritten as
  \be
\mathrm{Re}(S_{\mathrm{onshell}}) = \frac12\left(\sum_{a} J_{a}\ln J_{a}-J^{a} + \sum_{a\neq i<j} ( k_{ij}^{a}\ln  k_{ij}^{a}    - k_{ij}^{a}) \right)
\ee
which is easily recognized as the dominant\footnote{up to a term given by $  \frac14\ln(\prod_{a}( 2\pi J_{a}^{3}\prod_{i<j} (2\pi  k_{ij}^{a} ) )).$ }term in the stirling  expansion of  $\ln\sqrt{\prod_{a}(J_a+1)! \prod_{a\neq i<j}k_{ij}^{a}!)}$.  This cancels the factor in (\ref{normalized_20j}).  Let us now focus on the imaginary part:
\be\label{Ims}
{\mathrm{Im}}(S_{\mathrm{onshell}}) =  \frac12  \sum_{a\neq i<j} k_{ij}^{a} (\alpha_{i}^{aj} + \alpha_{j}^{ai}).
\ee
  First,  recall that the system of equations (\ref{system_of_equations}) possesses a gauge symmetry,
\be
\alpha^{ai}_{j} \to \alpha^{ai}_{j} +\theta^{ai}
\ee
where $\theta^{ai}= -\theta^{ia}$.
This  corresponds to the rotation of the  frame vector in the face $(ai)$ by an angle $\theta^{ai}$.
The action is invariant under these gauge transformations.  Indeed under $|z_{i}^{a}\ra \to e^{i\theta^{ai}} |z_{i}^{a}\ra $ the variation of the on-shell action is 
\bea
2\Delta S_{\mathrm{onshell}} &=& i \sum_{a\neq i<j} k_{ij}^{a} (\theta^{ai} + \theta^{aj}) =
i\sum_{(a,i,j)} k_{ij}^{a} \theta^{ai} = i\sum_{(a,i)} \left(\sum_{j\neq (a,i)}k_{ij}^{a}\right) \theta^{ai} \nn \\
&=& 2i \sum_{(a,i)} j_{ai} \theta^{ai} = i\sum_{(a,i)} j_{ai} ( \theta^{ai}+\theta^{ia}) =0
\eea
Here we have denoted by $(a,i,j)$ or $(a,i)$ a set of indices all distinct from each other. 
Therefore the on-shell action can be determined entirely in terms of gauge invariant angles.  The question is which combinations appear.

There are two types of gauge invariant data:
The first type characterizes the intrinsic geometry of each tetrahedron and depends only on the data associated with one tetrahedron.
These correspond  to the angles in a given tetrahedron $a$ between edges $(ij)$ and $(ik)$, and are given by
\be
\alpha^{ai}_{jk} \equiv \alpha^{ai}_{j}-\alpha^{ai}_{k}.
\ee
%This fact will be shown in the next section when we analyze these angles in more detail.
We already have seen in (\ref{eqn_alpha_i_jk}) that these angles are angles between the edges $(ij)$ and $(ik)$ at the tetrahedron $a$.

The second type of gauge invariant angles encode the extrinsic geometry of the gluing of the five tetrahedra. It depends on two tetrahedra and  involves the sum\footnote{Since the faces $(ab)$ and $(ba)$ have opposite orientations this is really a differences of angles when we take the orientation into account.} of angles between two tetrahedra
\be
\xi^{ab}_{i} \equiv \alpha^{ab}_{i} +\alpha^{ba}_{i}.
\ee

In order to understand the geometrical meaning of these angles let us first remark that when the shapes of the triangles $(ab)$ and $(ba)$ match then the angles between the edges of the triangles when viewed from $a$ or $b$ coincide.  Hence 
 \be \label{eqn_shape_match}
 \alpha^{ab}_{ij} =\alpha^{ba}_{ji}.
 \ee
This condition of shape matching therefore implies that
$
  0 = \alpha^{ab}_{ij} - \alpha^{ba}_{ji} = \xi^{ab}_{i} - \xi^{ab}_{j}
$
and so $\xi^{ab}_{i}$ is {\it independent} of  $i$.  This is the condition which will allow us to interpret $\xi^{ab}_i$ as the 4d dihedral angle between tetrahedra $a$ and $b$.

When the face matching condition is not satisfied, the geometry is twisted in the sense of \cite{twisted} 
and $\xi^{ab}_{i}$ represent a generalization of dihedral angles to twisted geometry.
%Let us first remark that when the shapes of the triangles $(ab)$ and $(ba)$ coincides this imply that the angles between the edges of the triangles 
% view from $a$ or $b$  coincides, hence 
% \be
% \alpha^{ab}_{ij} =\alpha^{ba}_{ji}.
% \ee
%This condition of shape matching therefore imply that $\xi^{ab}_{i}$ is {\it independent} of  $i$.}.
Moreover, the on-shell action will therefore represent a generalization of the Regge action to twisted geometry.  %A precise connection between these two actions will be given shortly.

%One might recognize that (\ref{eqn_xi_theta}) is the relationship between the 3d and 4d dihedral angles of a spherical 4-simplex \cite{Bahr:2009qd}.  Indeed, if $\xi^{ab}_{i}$ is independent of $i$ we can define $\xi^{ab} = \xi_^{ab}_i$ and $\xi^{aa} = 0$ then by \cite{} $\xi^{ab}$ will be the 4d dihedral angles of a spherical 4-simplex if the gram matrix $\{\cos \xi^{ab} \}$ is positive definite.

Let us now express the on-shell action in terms of this data.

\begin{theorem} \label{thm_twisted_action}
  The generalization of the Regge action to twisted geometry is given by
\be\label{action}
S_\bbT =   \sum_{i<j} j_{ij} \xi^{ij} +  \sum_{a\neq i<j} k_{ij}^{a} \alpha_{ij}^{a}.
\ee
where 
\be\label{defang}
   \xi^{ij} \equiv \frac13  \sum_{k\neq (i,j)} \xi^{ij}_{k}, \quad   \alpha^{a}_{ij} \equiv  \frac16 \sum_{b\neq (i,j,a)} (\alpha^{ai}_{jb} +\alpha^{aj}_{ib}).
\ee 
\end{theorem} 
\begin{proof}
Lets first recall the expression (\ref{Ims}) for the imaginary part of 
 the on-shell action 
 \be \label{eqn_2I_action}
2{\mathrm{Im}}(S_{\mathrm{onshell}}) \equiv 2 I  = \sum_{a\neq i<j} k_{ij}^{a}  (\alpha_{j}^{ai} + \alpha_{i}^{aj})
=\sum_{(a,i,j)} k_{ij}^{a}  \alpha_{j}^{ai}.
\ee
where we denote by $(i,j)$, $(a,i,j)$ a set of index distinct from each other.
We now evaluate the sum using the symmetries $ k^{a}_{ij}= k^{a}_{ji}$ , $j_{ij}=j_{ji}$ and the relation 
$ \sum_{j} k^{a}_{ij} = 2 j_{ai}$.
\bea \label{eqn_action_proof}
  \sum_{(a,i,j)} k^{a}_{ij} \alpha^{a}_{ij} 
%  &=& \sum_{a \neq i \neq j \neq b} k^{a}_{ij} (\alpha^{ai}_{jb} + \alpha^{aj}_{ib} ), \nn \\
  &=& \frac16 \sum_{(a,i,j,b)} k^{a}_{ij} (\alpha^{ai}_{j} + \alpha^{aj}_{i} - \alpha^{ai}_{b} - \alpha^{aj}_{b} )
  =  \frac13 \sum_{(a,i,j,b)} k^{a}_{ij} (\alpha^{ai}_{j}  - \alpha^{ai}_{b}),  \\
  &=& \frac13 \sum_{(a,i,j)} k^{a}_{ij} \sum_{b\neq (a,i,j)}\left(\alpha^{ai}_{j} - \alpha^{ai}_{b}\right)  
  = 2I - \sum_{(a,i)} (\sum_{j\neq(a,i)} k^{a}_{ij}) ( \frac13\sum_{b\neq (a,i)} \alpha^{ai}_{b}),
 \nn  \\
 &=& 2I - 2\sum_{(a,i)} j_{ai}  ( \frac13\sum_{b\neq (a,i)} \alpha^{ai}_{b})
 = 2I - \sum_{(a,i)} j_{ai} \xi^{ai} \nn.
\eea
 as required.
\end{proof}
Here $\xi^{ij}$ measures the extrinsic curvature of the face $(ij)$ inside the 4-simplex.  It is a generalisation of the 
dihedral angle in the case of twisted geometry.  The angle $\alpha^{a}_{ij}$ is a geometrical angle\footnote{See equation (\ref{eqn_alpha_geo}).} associated with the edge $(ij)$ inside the tetrahedron $a$. 

The first term is a generalization of the Regge action while the second term defines a canonical phase for the intertwiners as was noted in \cite{B}.  See also \cite{FC}.

\subsection{Geometricity and 4d Dihedral angles}

In this section we will discuss the connection between the $\xi^{ij}$ angles and the 4d dihedral angles of a 4-simplex when shape matching is enforced.  To do so we first derive relations between the angles $\xi$ and $\theta$ from the 3-term relations.  Indeed, using the reality condition (\ref{eqn_reality_cond}) and $|z^{a}_{b}\ket\bra z^{a}_{b}| + |z^{a}_{b}][ z^{a}_{b}| = A_{ab}\one$ we have 
\bea
  \left[z^{a}_{i}|z^{a}_{b} \right] \bra z^{b}_{a} | z^{b}_{i}] - [z^{a}_{i}|z^{a}_{b}\ket [ z^{b}_{a} | z^{b}_{i}] = \epsilon_{ab}\epsilon_{ai}\epsilon_{bi} A_{ab} \bra z^{i}_{a}|z^{i}_{b}\ket \\
  \left[z^{a}_{i}|z^{a}_{b} \right] \bra z^{b}_{a} | z^{b}_{i}\ket - [z^{a}_{i}|z^{a}_{b}\ket [ z^{b}_{a} | z^{b}_{i}\ket = \epsilon_{ab}\epsilon_{ai}\epsilon_{ib} A_{ab} \bra z^{i}_{a}|z^{i}_{b}] 
\eea
which are given explicitly by
\bea
  c^{a}_{ib}s^{b}_{ai} - s^{a}_{ib} c^{b}_{ai} e^{i \xi^{ab}_{i}} &=& \epsilon_{ab}\epsilon_{ai}\epsilon_{bi} c^{i}_{ab} e^{i(\xi^{ib}_{a}+\xi^{ab}_i-\xi^{ai}_{b})/2}, \label{eqn_4d_3_term_1}\\
  c^{a}_{ib}c^{b}_{ai} - s^{a}_{ib} s^{b}_{ai} e^{i \xi^{ab}_{i}} &=& \epsilon_{ab}\epsilon_{ai}\epsilon_{ib} s^{i}_{ab} e^{i(-\xi^{ib}_{a}+\xi^{ab}_b-\xi^{ai}_{b})/2}, \label{eqn_4d_3_term_2}
\eea
where $c^{a}_{ij} \equiv \cos  \frac{\theta^{a}_{ij}}2$ and $ s^{a}_{ij}\equiv \epsilon_{ij}\sin  \frac{\theta^{a}_{ij}}2$.  Taking the difference of squares of equations $|(\ref{eqn_4d_3_term_1})|^2 - |(\ref{eqn_4d_3_term_2})|^2$ we get
\be \label{eqn_4d_dihedral}
 -\cos \theta^{a}_{ib}\cos \theta^{b}_{ai} - \epsilon_{ib} \epsilon_{ai} \sin \theta^{a}_{ib} \sin \theta^{b}_{ai} \cos \xi^{ab}_{i} = \cos \theta^{i}_{ab}.
\ee

%Another  To give a precise notion to these angles, first notice that the reality condition $|z^{j}_{i}\ket = \epsilon_{ij} |z^{i}_j]$ in (\ref{eqn_reality_cond}) implies that $N^{j}_{i} = - N^{i}_{j}$.  
%\be
%  A_{ij} N^{j}_{i} \cdot \sigma = |z^{j}_{i}\ket \bra z^{j}_{i}| - |z^{j}_{i}] [ z^{j}_{i}| =  |z^{i}_{j}] [ z^{i}_{j}| - |z^{i}_{j}\ket \bra z^{i}_{j}| = -A_{ij} N^{i}_{j} %\cdot \sigma
%\ee
%One can also see that this implies that $F^{i}_{j} = F^{j}_{i}$, i.e. the reference vectors of glued triangles are aligned.  
Another way to derive this relationship is to use the relations $N^{a}_{b} = -N^{b}_{a}$ and $F^{a}_{b} = F^{b}_{a}$ and an argument similar to the one leading to (\ref{eqn_alpha_i_jk}) to show that 
\be
  L^{a}_{bi} \cdot L^{b}_{ai} = \epsilon_{ib} \epsilon_{ai} |L^{a}_{bi}| \cdot |L^{b}_{ai}| \cos \xi^{ab}_{i}.
\ee
Then using the defintion $L^{a}_{bi} = A_{ab} A_{ai} N^{a}_{b} \times N^{a}_{i}$ one arrives at $(\ref{eqn_4d_dihedral})$.
%and from this it follows that these angles are related to the 3d dihedral angles by
%\be \label{eqn_xi_theta}
%  - \cos \xi^{ab}_{i} = \frac{\cos \theta^{i}_{ab} + \cos \theta^{b}_{ai} \cos \theta^{a}_{bi}}{\epsilon_{bi} \epsilon_{ai} \sin \theta^{a}_{bi} \sin \theta^{b}_{ai}}.
%\ee

In the twisted picture we have three different $\xi^{ab}_{i}$ for $i\neq a,b$ and $\xi^{ab}$ is their average.  In order for the tetrahedra to glue together into a geometrical 4-simplex we must impose the shape matching conditions (\ref{eqn_shape_match}).  We already noted that when these conditions are satisfied $\xi^{ab}_{i}$ is independent of $i$.  Then as shown in \cite{Bahr:2009qd} equation (\ref{eqn_4d_dihedral}) is the relationship between the 3d and 4d dihedral angles of a  4-simplex\footnote{Note that our convention $\theta^{a}_{ii}=0$ differs from the other convention $\theta^{a}_{ii}=\pi$.}.  %Thus in this case $\xi^{ab}$ is indeed the 4d dihedral angle between tetrahedra $a$ and $b$.

Finally, we note that all the gauge invariant angles are entirely determined by the values of $k_{ij}^{a}$.
First the 3d dihedral angles are determined by the $k_{ij}^{a}$ via (\ref{eqn_k_z})
\be
\left(\sin \frac{\theta_{ij}^{a}}{2}\right)^{2} = \frac{J^{a} k_{ij}^{a}}{4 j_{i}^{a}j_{j}^{a}} 
%k_{ij}^{a},\qquad \cos \frac{\theta_{ij}^{a}}{2} = 
\ee
and then $\alpha^{ai}_{jk}$ and $\xi^{ab}_{i}$ are related to $\theta^{a}_{ij}$ by (\ref{eqn_alpha_theta}) and (\ref{eqn_4d_dihedral}) respectively.  Furthermore, these relations give an interpretation of $k^{a}_{ij}$ in terms of spherical geometry.  %This also shows that $k^{a}_{ij}$ can be interpretted as the square of the chordal distance between the vertices of a spherical tetrahedron.

\section{Discussion}

We introduced a new basis of SU(2) intertwiners which had the advantage of being both discrete and coherent.  Consequently, this basis was found to enjoy many of the advantages of both the orthogonal and coherent intertwiner bases.  

This basis was found to be overcomplete and satisfied a number of linear relations generated by the Pl\"ucker relation.  Using these relations we were able to explicitly derive the scalar products of these states in the 4-valent case.  Interestingly, it was found that these states could also be derived from an auxillary Hilbert space of states, not satisfying these linear relations, by projecting onto the kernel of an operator imposing these constraints.  

In the 4-valent case it was found that the orthogonal basis elements could be derived from the new basis by simply summing over one of the two spins labelling it.  This relationship allowed us to generate the various 15j symbols from a new, more fundamental symbol which we called the 20j symbol.  We also developed a generalized Racah formula for spin network amplitudes on an arbitrary graph.  In the case of a 4-simplex this provided another way of defining the $\{20j\}$ symbol.  %This result is expected to lead to the explicit derivation of amplitudes related to the Pachner moves.

Remarkably, it was found that the closed spin network amplitudes associated with this new basis admit an interpretation in terms of twisted geometries.  In particular we showed that the asymptotics of the 20j implies a classical action for twisted geometry.  Moreover, by studying the asymptotics of the 20j symbol, we were able to derive the asymptotics of the 15j symbol for the first time.  When the shapes of glued triangles were constrained to match, the action was found to reduce to the Regge action (with a canonical phase for the intertwiners).  When these constraints are not satisfied, the geometry is twisted and a generalization of the 4d dihedral angles was found in this case.  This agrees with the analysis given in \cite{Dittrich:2008ar} and \cite{Dittrich:2010ey}.  

What was not determined here is a four dimensional picture when the shape matching conditions do not hold.  A proposal for the Levi-Civita connection in twisted geometry was given in \cite{Haggard:2012pm} which was found to reduce to the usual spin connection in the Regge case.  It would be interesting to compare these results with the classical twisted geometry picture we found here.

\acknowledgments

We would like to thank Bianca Dittrich, Etera Livine, Valentin Bonzom, Eugenio Bianchi, Daniel Terno, and Carlo Rovelli for useful discussions.  JH would like to thank the Natural Sciences and Engineering Research Council of Canada (NSERC) for his post graduate scholarship. Research at Perimeter Institute is supported
by the Government of Canada through Industry Canada and by the Province of Ontario
through the Ministry of Research and Innovation.

\appendix

\section{Proof of Theorem \ref{thm_completeness}} \label{app_completeness}
\begin{proof}
To show this we introduce the following generating functional which depends holomorphically on $n$ spinors $|z_{i}\ket$ and $n(n-1)/2$ complex numbers $\tau_{ij} =-\tau_{ji}$
\be\label{defC}
 {\cal C}_{\tau_{ij}}(z_i)  \equiv  \sum_{[k] } 
 \prod_{i<j} \tau_{ij}^{k_{ij}} (z_i|k_{ij}\ket = e^{\sum_{i<j} \tau_{ij}[z_{i}|z_{j}\ket },
\ee
Here the sum is over all non-negative integers $[k]$ and not just those satisfying (\ref{kj}).  This functional was first considered by Schwinger \cite{Schwinger}.  The scalar product between the generating functionals is
\be\label{CC2}
\left\bra {\cal C}_{\tau_{ij}} | {\cal C}_{\tau_{ij}} \right\ket =  \int \prod_{i}\rd\mu(z_{i})   \left|{\cal C}_{\tau_{ij}}(z_{i})\right|^{2}
 =   \int \prod_{i}\rd\mu(z_{i}) e^{\sum_{i<j} \tau_{ij}[z_{i}|z_{j}\ket + \bar{\tau}_{ij} \bra z_{j}|z_{i}]}.
\ee
This integral is Gaussian and can be shown to have the following exact evaluation (for more details see \cite{FH_exact})
\be \label{eqn_scalar_det}
\left\bra {\cal C}_{\tau_{ij}} | {\cal C}_{\tau_{ij}} \right\ket
= \frac{1}{\det(1 + T\overline{T})}
\ee
where $T = (\tau_{ij})$ and $\overline{T} = (\overline{\tau}_{ij})$ are $n \times n$ antisymmetric matrices.  This determinant can be evaluated explicitly and in the case $n=4$ it has the form 
\be\label{det4}
\det(1 + T\overline{T}) = \left(1-\sum_{i<j} |\tau_{ij}|^{2} +  |R(\tau)|^{2}\right)^{2}
\ee
where $R(\tau)= \tau_{12}\tau_{34} +  \tau_{13} \tau_{42}+ \tau_{14}\tau_{23}$.  Notice that $R(\tau)$ is the Pl\"ucker identity and when $\tau_{ij} =[z_{i}|z_{j}\ket$ it vanishes.  Now by expanding the LHS of (\ref{CC2}) in the notation of (\ref{eqn_ST_notation})
\bea \label{eqn_scalar_expanded} 
 \left\bra {\cal C}_{\tau_{ij}} | {\cal C}_{\tau_{ij}} \right\ket&= & \sum_{j_i,S,T} \sum_{j'_i,S',T'} \left\bra S,T \right|\left. S',T' \right\ket \prod_{i<j} \bar{\tau}_{ij}^{k_{ij}} \tau_{ij}^{k_{ij}'}
 \eea
we see that the generating functional contains information about the scalar products of the new intertwiner basis.  
We now have two different ways to evaluate the scalar product for the generating functional 
with $\tau_{ij} = [z_i|z_j\ket$.
On one hand we start from (\ref{eqn_scalar_expanded}) to get 
\bea
 \left\bra {\cal C}_{[z_i|z_j\ket} | {\cal C}_{[z_i|z_j\ket} \right\ket&= &  \sum_{j_i,S,T} \sum_{j'_i,S',T'}(J+1)!^{2} \frac{( z_i | S,T \ket \bra S,T | S',T' \ket \bra S', T' | z_i )}{\|S,T\|^2 \|S',T'\|^2}. 
\eea
here we used the definition of our states and normalization:
\be
 \prod_{i<j} [z_i|z_j\ket^{k_{ij}} = (J+1)!\frac{ ( z_i | S,T \ket}{\|S,T\|^2}.
\ee
On the other hand we can evaluate directly  the product by expanding (\ref{det4}) when $R(\tau)=0$.
This gives
\be
 \left\bra {\cal C}_{[z_i|z_j\ket} | {\cal C}_{[z_i|z_j\ket} \right\ket =
 \sum_{j_i,S,T}(J+1)!^{2}\, \frac{( z_i | S,T \ket \bra S, T | z_i )}{\|S,T\|^2}.
\ee
equating the two expressions gives the identity decomposition
\be
\sum_{S',T'} \frac{\bra S,T | S',T' \ket \bra S', T' | z_i )}{ \|S',T'\|^2} = \bra S, T | z_i ).
\ee
This completes the proof.
\end{proof}
Note that the determinant in (\ref{eqn_scalar_det}) can be evaluated for general $n$ and the projector onto the space of $n$-valent intertwiners of spins $j_i$ has the same form as in (\ref{eqn_res_id}) (see \cite{FH_exact}).

\section{Proof of lemmas \ref{product1} and \ref{product2} } \label{alpha_proof}
In this appendix we show that
\begin{proposition}
The scalar product  (\ref{scalar})is given by 
\bea
  \left\bra S,T \right|\left. S',T' \right\ket &=&
\sum_{j_i,s,t,N}  \frac{(-1)^N}{N!} \frac{(J-N+1)!}{\prod_{i<j} k_{ij}(j_i-N/2,s,t)!} R^{(s,t)}_{(S,T)}(N) R^{(s,t)}_{(S',T')}(N) \\
&=& ( S,T | \Pi_{j_{i}} |  S',T' ).
\eea
\end{proposition}
\begin{proof}
To prove this we will use start from the computation of the scalar product of the generating functionals (\ref{eqn_scalar_expanded}) and its evaluation (\ref{det4}) which reads
\be \label{eqn_ST_loops}
  \sum_{j_i,S,T} \sum_{j'_i,S',T'} \left\bra S,T \right|\left. S',T' \right\ket \prod_{i<j} \bar{\tau}_{ij}^{k_{ij}(j_i,S,T)} \tau_{ij}^{k_{ij}(j_i,S',T')} = \left( 1 - \sum_{i < j} |\tau_{ij}|^2 + |R(\tau)|^2 \right)^{-2},
\ee
where $R(\tau) = \tau_{12}\tau_{34} + \tau_{13}\tau_{42} + \tau_{14}\tau_{23}$.  Expanding the RHS of (\ref{eqn_ST_loops}) gives
\be
 \sum_{[k],N}  \frac{(-1)^N(J+N+1)! }{N !}  |R(\tau)|^{2N} \prod_{i<j} \frac{|\tau_{ij}|^{2k_{ij}}}{k_{ij}!}.
\ee
and by shifting $j_i \rightarrow j_i - N/2$ and using the relations (\ref{rel}) this becomes
\be
  \sum_{j_i,s,t,N}  \frac{(-1)^N(J-N+1)!}{N!} \sum_{S,T} \sum_{S',T'} R^{(s,t)}_{(S,T)}(N) R^{(s,t)}_{(S',T')}(N) \prod_{i<j} \frac{\tau_{ij}^{k_{ij}(j_i,S,T)}{\bar{\tau}}_{ij}^{k_{ij}(j_i,S',T')}}{k_{ij}(j_i-N/2,s,t)!} 
\ee
Now comparing coefficients of this with the LHS of (\ref{eqn_ST_loops}) gives the desired result, that is 
\be
\left\bra S,T \right|\left. S',T' \right\ket=
\sum_{j_i,s,t,N}  \frac{(-1)^N}{N!} \frac{(J-N+1)!}{\prod_{i<j} k_{ij}(j_i-N/2,s,t)!} R^{(s,t)}_{(S,T)}(N) R^{(s,t)}_{(S',T')}(N).
\ee
This can also be written as 
\be
\left\bra S,T \right|\left. S',T' \right\ket=
\sum_{j_i,s,t,N}  \frac{(-1)^N}{N!}||s,t||_{j_{i}-N/2}^{2 }  \frac{(J-N+1)!}{(J-2N +1)!} R^{(s,t)}_{(S,T)}(N) R^{(s,t)}_{(S',T')}(N).
\ee
 Using the expression (\ref{matrixel}) for the matrix elements of $\hat{R}^{N}$:
 \be
{}_{j_{i}} (S,T| \hat{R}^{N}  | s,t)_{j_{i}-N/2} = ||s,t||_{j_{i}-N/2}^{2 }  \frac{(J+1)!}{(J-2N +1)!} 
  R^{(s,t)}_{(S,T)}(N),
\ee
this expression reads
\bea
%( S,T | \Pi_{j_{i}} |  S',T' ) =
&&\sum_{N,s,t}\frac{(-1)^{N}}{N!} \frac{(J-N+1)!(J-2N+1)!}{(J+1)!^{2}} \, 
\frac{( S,T |\hat{R}^{N}|s,t)_{j_{i}-N/2}(s,t|(\hat{R}^{\dagger})^{N} |S',T')}{||s,t||_{j_{i}-N/2}^{2}}\nn\\
&=& \sum_{N=0}^{\mathrm{min}(2j_{i})}\frac{(-1)^{N}}{N!} \frac{(J-N+1)!(J-2N+1)!}{(J+1)!^{2}} \,\nn 
( S,T |\hat{R}^{N}(\hat{R}^{\dagger})^{N} |S',T')\\
&=& ( S,T | \Pi_{j_{i}} |  S',T' ).\nn
\eea
where we have used the decomposition of the identity in the second line.
%\end{proof}
%\be
%\left\bra S,T  \right.\left|S',T'  \right \ket =
%\delta_{j_i,j'_i} \sum_{s,t,N} \left( \frac{(-1)^N (J-N+1)!}{N!\prod_{i<j} k_{ij}(j_i-N/2,s,t)!} \right) \,R^{(s,t)}_{(S,T)}(N) R^{(s,t)}_{(S',T')}(N)
%\ee 
%which shows that
%\be 
% \alpha_{j_{i},N}^{(s,t)} = \frac{(-1)^N (J-N+1)!}{N!\prod_{i<j} k_{ij}(j_i-N/2,s,t)!}.
%\ee
\end{proof}
\section{Proof of Equation (\ref{eqn_R_delta})} \label{app_R_delta}

Writing $a = N+s-S$ and $b=N+t-T$ with $a,b\geq 0$ and $a+b\leq N$ we get
\bea
\sum_{T} R^{(s,t)}_{(S,T)}(N) = \sum_{b=0}^{N-a} \frac{(-1)^{b}N!}{a!b!(N-a-b)!} = \sum_{b=a}^{N} (-1)^{b-a} \binom{N}{b} \binom{b}{a} = \delta_{N,a} = \delta_{S,s},
\eea
where we used a standard binomial identity in the last step.  Similarly we use the same identity to show that
\bea
\sum_{S} (-1)^{k_{23}(j_i,S,T)} R^{(s,t)}_{(S,T)}(N) = (-1)^{k_{23}(j_i,s,t)} \sum_{a=0}^{N-b} \frac{(-1)^{a}N!}{a!b!(N-a-b)!} = (-1)^{k_{23}(j_i,s,t)} \delta_{T,t}.
\eea

\section{Resummations of scalar products} \label{sec_scalar_products}

\begin{proposition}
The summations in (\ref{eqn_usual_sc_prods}) can be performed explicitly.  The first of which has the form
\be
\left\bra S\right|\left.S' \right\ket  =\frac{\delta_{S,S'}}{2S+1}  \Delta^{2}(j_{1}j_{2}S) \Delta^{2}(j_{3}j_{4}S)
\ee
where $\Delta(j_{1}j_{2}j_{3})$ is defined in (\ref{eqn_tri_coeff}).
\end{proposition}

\begin{proof}
We wish to perform the summation
\be
  \left\bra S\right|\left.S' \right\ket = \delta_{S,S'} \sum_{t,N} \frac{(-1)^N (J-N+1)!}{N!\prod_{i<j} k_{ij}(j_i-N/2,s,t)!}.
\ee
We can first evaluate the sum over $t$ by noticing that
\be
  \sum_{T} \frac{1}{\prod_{i<j} k_{ij}!} = \frac{1}{k_{12}! k_{34}! (k_{13}+k_{14})!(k_{23}+k_{24})!} \sum_{T} \binom{k_{13}+k_{14}}{k_{13}}\binom{k_{23}+k_{24}}{k_{23}}
\ee
and using Vandermonde's identity\footnote{For proof compare the coefficients in the expansion of $(1+x)^p(1+x)^q = (1+x)^{p+q}$.}
\be \label{vandermonde}
  \sum_{k} \binom{p}{k} \binom{q}{j-k} = \binom{p+q}{j}
\ee
we have
\be
  \sum_{T} \frac{1}{k_{ij}!} = \frac{(2S)!}{k_{12}!k_{34}!k_{1}!k_{2}!k_{3}!k_{4}!}.
\ee
where $k_{1} = k_{13} + k_{14} = j_1 - j_2 + S$, $k_{2} = k_{23} + k_{24} = j_2 - j_1 + S$, $k_{3} = k_{13} + k_{23} = j_3 - j_4 + S$, and $k_{4} = k_{14} + k_{24} = j_4 - j_3 + S$.  Therefore after changing the  variable $j_i$ to $j_{i}-N/2$ we have
\be
   \left\bra S\right|\left.S' \right\ket = \delta_{S,S'} \sum_{N} \frac{(-1)^N (J-N+1)!(2S)!}{N!(k_{12}-N)!(k_{34}-N)!k_{1}!k_{2}!k_{3}!k_{4}!}.
\ee
We can now perform the sum over $N$
\be
  \sum_N (-1)^{N} \frac{(J-N+1)!}{N!(k_{12}-N)!(k_{34}-N)!} = \frac{(j_1+j_2+S+1)!}{k_{12}!}\sum_N (-1)^{N} \binom{k_{12}}{N} \binom{J-N+1}{j_1+j_2+S+1}
\ee
using the identity\footnote{For proof compare the coefficients in the expansion of $(1+x)^p/(1+x)^{q+1} = (1+x)^{p-q -1}$ with $p < q$.}
\be \label{binom_id}
  \sum_{k} (-1)^{k} \binom{p}{k} \binom{j+q- k}{q} = \binom{j+q-p}{j}
\ee
this becomes
\be
  \sum_N (-1)^{N} \frac{(J-N+1)!}{N!(k_{12}-N)!(k_{34}-N)!} = \frac{(j_1+j_2+S+1)!(j_3+j_4+S+1)!}{k_{12}! k_{34}!(2S+1)!}
\ee
and finally
 \be
  \bra S | S' \ket
  = \delta_{S,S'} \frac{(j_1+j_2+S+1)!(j_3+j_4+S+1)!}{(2S+1)k_{12}!k_{34}!k_{1}!k_{2}!k_{3}!k_{4}!}.
\ee
\end{proof}

\section{Racah formula for the 20j symbol} \label{20j_symbol}

In this section we give an explicit parameterization of the 37 $M_C$ in terms of the 17 parameters $p_k$.  Choosing the following parameters: $p_{1} = M_{1324}, p_{2} = M_{1325}, p_{3} = M_{1345}, p_{4} = M_{1354}, p_{5} = M_{1435}, p_{6} = M_{1425}, p_{7} = M_{2345}, p_{8} = M_{2354}, p_{9} = M_{2435}, p_{10} = M_{12345}, p_{11} = M_{12543}, p_{12} = M_{13245}, p_{13} = M_{13254}, p_{14} = M_{13425}, p_{15} = M_{13524}, p_{16} = M_{14235}, p_{17} = M_{14325}$ we can solve for the $M_C$ in (\ref{eqn_kee}) to be
\bea
  M_{123} &=& k^{3}_{12}-p_{1}-p_{2}-p_{12}-p_{13}, \hspace{12pt} M_{145} = k^{1}_{45}-p_{6}-p_{5}-p_{16}-p_{17}, \nonumber \\
  M_{124} &=& k^{4}_{12}-p_{1}-p_{6}-p_{16}-p_{15}, \hspace{12pt} M_{234} = k^{2}_{34}-p_{1}-p_{8}-p_{12}-p_{16}, \nonumber \\
  M_{125} &=& k^{5}_{12}-p_{2}-p_{6}-p_{14}-p_{17}, \hspace{12pt} M_{235} = k^{2}_{35}-p_{2}-p_{7}-p_{13}-p_{17}, \nonumber \\
  M_{134} &=& k^{1}_{34}-p_{1}-p_{4}-p_{13}-p_{15}, \hspace{12pt} M_{245} = k^{2}_{45}-p_{9}-p_{6}-p_{14}-p_{15}, \nonumber \\
  M_{135} &=& k^{1}_{35}-p_{2}-p_{3}-p_{12}-p_{14}, \hspace{12pt} M_{345} = k^{4}_{35}-p_{7}-p_{3}-p_{10}-p_{11}, \nonumber \\
  M_{1234} &=& k^{3}_{24}-k^{2}_{34}+p_{1}+p_{8}+p_{12}+p_{16}-p_{7}-p_{10}-p_{17}, \nonumber \\
  M_{1243} &=& k^{3}_{14}-k^{1}_{34}+p_{1}+p_{4}+p_{13}+p_{15}-p_{3}-p_{14}-p_{11}, \nonumber \\
  M_{1245} &=& k^{5}_{14}-k^{1}_{45}+p_{6}+p_{5}+p_{16}+p_{17}-p_{3}-p_{10}-p_{12}, \nonumber \\           
  M_{1254} &=& k^{5}_{24}-k^{2}_{45}+p_{9}+p_{6}+p_{14}+p_{15}-p_{7}-p_{11}-p_{13}, \nonumber \\
  M_{12354} &=& k^{4}_{15}-k^{1}_{45}+k^{2}_{45}-k^{5}_{24}+p_{7}+p_{5}+p_{11}+p_{16}+p_{17}-p_{4}-p_{14}-p_{15}-p_{9}, \nonumber \\
  M_{12453} &=& k^{4}_{25}-k^{2}_{45}+k^{1}_{45}-k^{5}_{14}+p_{9}+p_{3}+p_{10}+p_{14}+p_{15}-p_{5}-p_{8}-p_{16}-p_{17}, \nonumber \\
  M_{12435} &=& k^{4}_{23}-k^{2}_{34}+k^{1}_{34}-k^{3}_{14}+p_{8}+p_{3}+p_{12}+p_{16}+p_{11}-p_{4}-p_{9}-p_{13}-p_{15}, \nonumber \\
  M_{12534} &=& k^{4}_{13}-k^{1}_{34}+k^{2}_{34}-k^{3}_{24}+p_{4}+p_{7}+p_{13}+p_{15}+p_{10}-p_{5}-p_{8}-p_{12}-p_{16}, \nonumber \\
  M_{1235} &=& k^{3}_{25}-k^{2}_{35}+k^{1}_{45}-k^{4}_{15}+k^{5}_{24}-k^{2}_{45}+p_{2}+p_{13}+p_{9}+p_{4}+p_{14}+p_{15}-p_{8}-p_{5}-p_{11}-2p_{16}, \nonumber \\
  M_{1253} &=& k^{5}_{23}-k^{2}_{35}+k^{1}_{34}-k^{4}_{13}+k^{3}_{24}-k^{2}_{34}+p_{8}+p_{5}+p_{12}+p_{16}+p_{2}+p_{17}-p_{9}-p_{4}-p_{10}-2p_{15}. \nonumber 
\eea          
where $k^{i}_{jk}$ are parameterized in terms of $S_i$ and $T_i$ as in (\ref{int1}) and (\ref{int2}) by the relations $2j_{ij} = \sum_{k} k^{j}_{ik}$ and $2j_{jk} = \sum_{i} k^{j}_{ik}$.

%%%%%%%%%%%%%%%%%%%%%%%%

\end{document}